\documentclass[
  submission, % change to adraft or preliminary
  copyright % xor publicdomain
]{eptcs}
\providecommand{\event}{QPL 2026}

\usepackage{iftex}

\ifpdf
  \usepackage{underscore}         % Only needed if you use pdflatex.
  \usepackage[T1]{fontenc}        % Recommended with pdflatex
\else
  \usepackage{breakurl}           % Not needed if you use pdflatex only.
\fi

\usepackage[english]{babel} %dealing with underscore package issues

\usepackage{todonotes}[disable] % disable for submission
\usepackage{amsthm}
\usepackage{enumitem}
\usepackage{multirow} % rowspan command: \multirow
\usepackage{subcaption}
\usepackage{amsmath}
\usepackage{thmtools} % apparently needed for cleverref to work with amsthm, see https://tex.stackexchange.com/a/732209
\usepackage{amssymb}
\usepackage[]{tikzit} % refers to tikzit-quantum/tikzit.sty
\tikzstyle{gate}=[shape=rectangle, text height=1.5ex, text depth=0.25ex, yshift=-0.5mm, fill=white, draw=black, minimum height=5mm, minimum width=5mm, font={\small}, tikzit category=circuit]
\tikzstyle{big gate}=[shape=rectangle, text height=1.5ex, text depth=0.25ex, yshift=-0.5mm, fill=white, draw=black, minimum height=10mm, minimum width=5mm, font={\small}, tikzit category=circuit]
\tikzstyle{Z dot}=[inner sep=0mm, minimum size=2mm, shape=circle, draw=black, fill={rgb,255: red,221; green,255; blue,221}, tikzit category=zx]
\tikzstyle{Z phase dot}=[minimum size=5mm, font={\footnotesize\boldmath}, shape=rectangle, rounded corners=2mm, inner sep=0.2mm, outer sep=-2mm, scale=0.8, tikzit shape=circle, draw=black, fill={rgb,255: red,221; green,255; blue,221}, tikzit draw=blue, tikzit category=zx]
\tikzstyle{X dot}=[Z dot, shape=circle, draw=black, fill={rgb,255: red,255; green,136; blue,136}, tikzit category=zx]
\tikzstyle{X phase dot}=[Z phase dot, tikzit shape=circle, draw=black, tikzit draw=blue, fill={rgb,255: red,255; green,136; blue,136}, font={\footnotesize\boldmath}, tikzit category=zx]
\tikzstyle{hadamard}=[fill=yellow, draw=black, shape=rectangle, inner sep=0.6mm, minimum height=1.5mm, minimum width=1.5mm, tikzit category=zx]
\tikzstyle{paulibox}=[fill={rgb,255: red,221; green,221; blue,255}, draw=black, shape=rectangle, inner sep=0.6mm, minimum height=5mm, minimum width=5mm, font={\footnotesize}, text height=1.5ex, text depth=0.25ex, tikzit category=zx]
\tikzstyle{vertex}=[inner sep=0mm, minimum size=1mm, shape=circle, draw=black, fill=black, tikzit category=misc]
\tikzstyle{vertex set}=[inner sep=0mm, minimum size=1mm, shape=circle, draw=black, fill=white, font={\footnotesize\boldmath}, tikzit category=misc]
\tikzstyle{small black dot}=[fill=black, draw=black, shape=circle, inner sep=0pt, minimum width=1.2mm, tikzit category=circuit]
\tikzstyle{cnot ctrl}=[fill=black, draw=black, shape=circle, inner sep=0pt, minimum width=1.2mm, tikzit category=circuit]
\tikzstyle{cnot targ}=[fill=white, draw=white, shape=circle, tikzit category=circuit, label={center:$\oplus$}, inner sep=0pt, minimum width=2.1mm, tikzit fill={rgb,255: red,102; green,204; blue,255}, tikzit draw=black]
\tikzstyle{ket}=[fill=white, draw=black, shape=regular polygon, regular polygon sides=3, regular polygon rotate=-30, scale=0.7, inner sep=1pt, tikzit category=circuit, tikzit shape=rectangle, tikzit fill=green]
\tikzstyle{bra}=[fill=white, draw=black, shape=regular polygon, regular polygon sides=3, regular polygon rotate=30, scale=0.7, inner sep=1pt, tikzit category=circuit, tikzit shape=rectangle, tikzit fill=red]
\tikzstyle{scalar}=[shape=rectangle, text height=1.5ex, text depth=0.25ex, yshift=-0.5mm, fill=white, draw=black, minimum height=5mm, minimum width=5mm, font={\small}]
\tikzstyle{clabel}=[fill=white, draw=none, shape=rectangle, tikzit fill={rgb,255: red,56; green,255; blue,242}, font={\footnotesize}, inner sep=1pt, tikzit category=labels]
\tikzstyle{empty diagram}=[draw={gray!40!white}, dashed, shape=rectangle, minimum width=1cm, minimum height=1cm, tikzit category=misc]
\tikzstyle{amap}=[fill=white, draw=black, shape=NEbox, tikzit category=asymmetric, tikzit fill=yellow, tikzit shape=rectangle]
\tikzstyle{amap conj}=[fill=white, draw=black, shape=NWbox, tikzit category=asymmetric, tikzit fill=green, tikzit shape=rectangle]
\tikzstyle{amap adj}=[fill=white, draw=black, shape=SEbox, tikzit category=asymmetric, tikzit fill=red, tikzit shape=rectangle]
\tikzstyle{amap trans}=[fill=white, draw=black, shape=SWbox, tikzit category=asymmetric, tikzit fill=orange, tikzit shape=rectangle]
\tikzstyle{astate}=[fill=white, draw=black, shape=NEtriangle, tikzit category=asymmetric, tikzit shape=circle, tikzit fill=yellow]
\tikzstyle{astate conj}=[fill=white, draw=black, shape=NWtriangle, tikzit category=asymmetric, tikzit shape=circle, tikzit fill=green]
\tikzstyle{astate adj}=[fill=white, draw=black, shape=SEtriangle, tikzit category=asymmetric, tikzit shape=circle, tikzit fill=red]
\tikzstyle{astate trans}=[fill=white, draw=black, shape=SWtriangle, tikzit category=asymmetric, tikzit shape=circle, tikzit fill=orange]
\tikzstyle{tri}=[fill={rgb,255: red,255; green,136; blue,136}, tikzit fill=red, draw=black, shape=regular polygon, regular polygon sides=3, regular polygon rotate=30, scale=0.7, inner sep=0.2pt, minimum size=7mm, tikzit shape=rectangle]
\tikzstyle{tri left}=[fill={rgb,255: red,255; green,136; blue,136}, tikzit fill=red, draw=black, shape=regular polygon, regular polygon sides=3, regular polygon rotate=210, scale=0.7, inner sep=0.2pt, minimum size=7mm, tikzit shape=rectangle]
\tikzstyle{hadamard edge}=[-, dashed, dash pattern=on 2pt off 0.5pt, thick, draw={rgb,255: red,68; green,136; blue,255}]
\tikzstyle{box edge}=[-, dash pattern=on 2pt off 0.5pt, thick, draw={rgb,255: red,203; green,192; blue,225}, dashed]
\tikzstyle{brace edge}=[-, tikzit draw=blue, decorate, decoration={brace,amplitude=1mm,raise=-1mm}]
\tikzstyle{diredge}=[->]
\tikzstyle{double edge}=[-, double, shorten <=-1mm, shorten >=-1mm, double distance=2pt]
\tikzstyle{gray edge}=[-, {gray!60!white}]
\tikzstyle{pointer edge}=[->, very thick, gray]
\tikzstyle{boldedge}=[-, line width=1.6pt, shorten <=-0.17mm, shorten >=-0.17mm]
\tikzstyle{bidir edge}=[<->, very thick, draw={rgb,255: red,191; green,191; blue,191}]

\input{tikzit-quantum/quantum.tikzdefs}
\tikzstyle{Z dot}=[inner sep=0mm, minimum size=2mm, shape=circle, draw=black, fill=zxGreen, tikzit fill={rgb,255: red,216; green,248; blue,216}, outer sep=-0.5mm]
\tikzstyle{Z phase dot}=[draw=black, fill=zxGreen, shape=rectangle, minimum size=4.5mm, rounded corners=1.8mm, inner sep=0.5mm, outer sep=-0.5mm, scale=0.8, tikzit shape=circle, tikzit fill={rgb,255: red,216; green,248; blue,216}, font={\footnotesize\boldmath}]
\tikzstyle{X dot}=[shape=circle, draw=black, fill=zxRed, tikzit fill={rgb,255: red,221; green,165; blue,165}, inner sep=0 mm, minimum size=2 mm, outer sep=-0.5mm]
\tikzstyle{X phase dot}=[Z phase dot, draw=black, fill=zxRed, tikzit fill={rgb,255: red,221; green,165; blue,165}]
\tikzstyle{H box}=[fill=zxHad, draw=black, shape=rectangle, inner sep=0.6mm, minimum height=1.5mm, minimum width=1.5mm, tikzit fill=yellow, font={\footnotesize\boldmath}]
\tikzstyle{box}=[draw=black, shape=rectangle, fill=white, minimum size=1em, inner sep=0.2em, scale=0.85, font={\scriptsize}, outer sep=-0.5mm]
\tikzstyle{black dot}=[fill=black, draw=black, shape=circle, inner sep=1pt]
\tikzstyle{sLabel}=[font={\scriptsize}, tikzit draw=black, auto]
\tikzstyle{not}=[draw=black, circle, addcross, minimum size=2mm, outer sep=-0.5mm, inner sep=0mm]
\tikzstyle{Z dot thick}=[inner sep=0mm, minimum size=2mm, shape=circle, draw=black, fill=zxGreen, tikzit fill={rgb,255: red,216; green,248; blue,216}, outer sep=-0.5mm, line width=1pt]
\tikzstyle{X dot thick}=[inner sep=0mm, minimum size=2mm, shape=circle, draw=black, fill=zxRed, tikzit fill={rgb,255: red,221; green,165; blue,165}, outer sep=-0.5mm, line width=1pt]
\tikzstyle{fault-location}=[fill=white, draw=black, shape=circle, minimum size=2mm, inner sep=0mm, outer sep=-0.5 mm, regular polygon, regular polygon sides=8]
\tikzstyle{fault-location-faulty}=[fill=white, draw={rgb,255: red,191; green,0; blue,64}, shape=circle, minimum size=2mm, inner sep=0mm, outer sep=-0.5 mm, regular polygon, regular polygon sides=8, minimum size=3mm]
\tikzstyle{new style 0}=[fill=white, draw=black, shape=circle]

\tikzstyle{dashed-line}=[-, style=dashed, draw={rgb,255: red,128; green,128; blue,128}]
\tikzstyle{dotted-line}=[-, style=dotted, line width=0.5mm, line cap=round, dash pattern={on 0pt off 3\pgflinewidth}, draw=black]
\tikzstyle{classical}=[-, double, draw=black]
\tikzstyle{hadamard}=[-, style=dashed, draw=blue]
\tikzstyle{X Web}=[-, preaction={line width=1mm, draw=zxDarkRed}, tikzit draw=red]
\tikzstyle{z detector}=[-, zsideline, tikzit draw=green]
\tikzstyle{x detector}=[-, xsideline, tikzit draw=red]
\tikzstyle{II B II}=[-, boldwire, tikzit draw=gray]
\tikzstyle{IZ W II}=[->, ZL1, wire, tikzit draw=green]
\tikzstyle{IZ T II}=[->, ZL1, nowire, tikzit draw=green]
\tikzstyle{IZd T II}=[->, ZL1-dotted, nowire, tikzit draw=green]
\tikzstyle{ZI W II}=[->, ZL2, wire, tikzit draw=green]
\tikzstyle{ZI T II}=[->, ZL2, nowire, tikzit draw=green]
\tikzstyle{IX W II}=[->, XL1, wire, tikzit draw=red]
\tikzstyle{IX T II}=[->, XL1, nowire, tikzit draw=red]
\tikzstyle{IZ W ZI}=[->, ZL1, wire, ZR1, tikzit draw=green]
\tikzstyle{IZ B ZI}=[->, ZL1, boldwire, ZR1, tikzit draw=green]
\tikzstyle{IX W XI}=[->, XL1, wire, XR1, tikzit draw=red]
\tikzstyle{IX T XI}=[->, XL1, nowire, XR1, tikzit draw=red]
\tikzstyle{IZ W XI}=[->, ZL1, wire, XR1, tikzit draw=yellow]
\tikzstyle{ZZ T II}=[->, ZL2, ZL1, nowire, tikzit draw=green]
\tikzstyle{XX T II}=[->, XL2, XL1, nowire, tikzit draw=red]
\tikzstyle{ZZ W II}=[->, ZL2, ZL1, wire, tikzit draw=green]
\tikzstyle{XX W II}=[->, XL2, XL1, wire, tikzit draw=red]
\tikzstyle{ZZ W ZI}=[->, ZL2, ZL1, wire, ZR1, tikzit draw=green]
\tikzstyle{ZZ W XX}=[->, ZL2, ZL1, wire, XR1, XR2, tikzit draw=green]
\tikzstyle{Z Web}=[-, preaction={line width=1.7mm, draw=zxDarkGreen}, tikzit draw=green]
\tikzstyle{bridge}=[-, line width=2.1mm, draw=white, tikzit draw=green]
\tikzstyle{bridge 0}=[-, line width=1mm, draw=white, tikzit draw=gray]
\tikzstyle{Z Web T}=[-, nowire, preaction={line width=1.7mm, draw=zxDarkGreen}, tikzit draw=green]
\tikzstyle{XZ Web}=[-, preaction={line width=1.8mm, draw=zxDarkGreen}, preaction={line width=1mm, draw=zxDarkRed}, tikzit draw=blue]
\tikzstyle{braceedge}=[-, decorate, decoration={brace, amplitude=2mm, raise=-1mm}]
\tikzstyle{arrow}=[->]
\tikzstyle{fault-free}=[-, draw={rgb,255: red,177; green,98; blue,255}, line width=1pt]
\tikzstyle{new edge style 0}=[-, fill={rgb,255: red,211; green,211; blue,211}, draw=none]
\tikzstyle{new edge style 1}=[->, draw={rgb,255: red,177; green,98; blue,255}, line width=1pt]
\tikzstyle{fault-accounting-component}=[-, style=dashed, draw={rgb,255: red,0; green,128; blue,128}]
\tikzstyle{new edge style 2}=[-, style=dashed, draw={rgb,255: red,0; green,128; blue,128}]
\tikzstyle{X surface}=[-, fill=zxRed, tikzit fill=red, draw=none, tikzit draw=red]
\tikzstyle{Y surface}=[-, fill=zxGreen, tikzit fill=green, draw=none, tikzit draw=green]
\input{mp.tikzdefs}
\usepackage[capitalise]{cleveref} % TODO: is this the right loading order for cleverref?

\newtheorem{theorem}{Theorem}[section]
\newtheorem{proposition}[theorem]{Proposition}

\newtheorem{corollary}[theorem]{Corollary}
\newtheorem{definition}[theorem]{Definition}

\newtheorem{example}[theorem]{Example}

\title{Preserving MWPM-Decodability in Fault-Equivalent Rewrites}
\author{Maximilian Schweikart
\institute{University of Oxford\\Oxford, UK}
\email{maximilian.schweikart@kit.edu}
\and
Linnea Grans-Samuelsson
\institute{University of Oxford\\Oxford, UK}
\email{linnea.grans-samuelsson@physics.ox.ac.uk}
\and
Aleks Kissinger
\institute{University of Oxford\\Oxford, UK}
\email{aleks.kissinger@cs.ox.ac.uk}
\and
Benjamin Rodatz
\institute{University of Oxford\\Oxford, UK}
\email{benjamin.rodatz@cs.ox.ac.uk}
}

\newcommand{\titlerunning}{Preserving MWPM-Decodability in Fault-Equivalent
    Rewrites}
\newcommand{\authorrunning}{M. Schweikart, L. Grans-Samuelsson, A. Kissinger \& B. Rodatz}

\hypersetup{
  bookmarksnumbered,
  pdftitle    = {\titlerunning},
  pdfauthor   = {\authorrunning},
  pdfsubject  = {QPL},               % TODO: Consider adding a more appropriate subject or description
}

\begin{document}
\maketitle

\begin{abstract}
  Decoding a quantum error correction code is generally $NP$-hard, but corrections must be applied at a high frequency to suppress noise successfully.
  Matchable codes, like the surface code, exhibit a special structure that makes it possible to efficiently, approximately solve the decoding problem through minimum-weight perfect matching (MWPM).
  However, this efficiency-enabling property can be lost when constructing
  implementations for fault-tolerant gadgets such as syndrome-extraction circuits or logical operations.

  In this work, we take a circuit-centric perspective to formalise how the decoding problem changes when applying ZX rewrites to a ZX~diagram with a given detector basis.
  We demonstrate a set of rewrites that preserve MWPM-decodability of circuits and show that these \emph{matchability-preserving rewrites} can be used to fault-tolerantly extract quantum circuits from phase-free ZX diagrams. 
  In particular, this allows us to build efficiently decodable, fault-tolerant syndrome-extraction circuits for matchable codes.
\end{abstract}

\section{Introduction}
\label{sec:intro}

A key challenge in scalable and reliable quantum computing is the noise incurred by the computer's interaction with the environment \cite{quantum-computing-in-the-nisq-era-and-beyond}. 
Fault-tolerant quantum computing (FTQC) aims to overcome this challenge by encoding computation with redundancy such that errors can be detected and corrected \cite{gottesman-introduction-to-qec-and-ftqc}. 
This is achieved by performing measurements whose outcomes, in the noise-free case, are predetermined. 
Errors that flip measurement outcomes in such a way that our assumptions are violated are then detected by the circuit \cite{Derks2025designingfault}.
A key challenge in fault-tolerant quantum computing is taking these violated measurement outcomes and determining which error is the most likely to have caused it --- the decoding. 
In general, this problem is known to be $NP$-hard \cite{np-hardness-of-decoding-qec-codes}.
At the same time, decoding has to occur in real time during the quantum computation, on timescales that can be very short. As a consequence, an important problem is finding good algorithms that quickly approximate good solutions to the decoding problem \cite{pymatching-v1}. The feasibility of such approximation depends strongly on the decoding problem at hand.
As such, in FTQC there are two related problems: (1) finding circuits that are good at detecting errors and (2) finding efficient decoders for these circuits. 

Recent work \cite{fault-tolerance-by-construction,floquetifying-stabiliser-codes} has proposed using the ZX calculus for tackling the former problem of synthesising quantum circuits that are good at detecting errors. 
The authors propose starting with an idealised specification of the circuit one wants to implement and how it should behave under noise. 
While such specifications are relatively easy to find, they are not usually expressed in terms of implementable, noisy gates available on a quantum computer. 
The specification is then rewritten, using only rewrites that preserve the specification's behaviour under noise, until one gets a circuit that can be run on a quantum computer. 
By construction, as each step preserves the circuit's behaviour under noise, the resulting circuit is good at detecting errors. 
However, \cite{fault-tolerance-by-construction} solely focuses on designing circuits that are good at detecting errors without systematically considering the decoding problem.
Related work~\cite{poor2025ultralowoverheadsyndrome} studies the decoding of very specific circuits but not the general problem.

Circuit rewrites can change the number of linearly dependent measurement outcomes, i.e.\@ the detecting sets. 
As such, circuit rewrites can substantially change the underlying decoding problem. 
Therefore, if one starts with a specification that has an efficient decoder, the procedure proposed by \cite{fault-tolerance-by-construction} does not generally guarantee that the resulting implementation is still efficiently decodable. 
In this work, we consider the effect of rewrites on the decoding problem, defining a set of rewrites that not only preserve the specification's behaviour under noise but also its efficient decodability. 

There are two potential avenues for tackling this problem: (1) either one restricts oneself to rewrites that preserve the associated decoding problem or (2) one preserves sufficiently many properties about the decoding problem such that efficient decodability is preserved. 
\cite{fault-tolerant-transformations-of-spacetime-codes} take the former avenue, restricting the set of rewrites to preserve the decoding problem. 
While powerful, this is a very restrictive choice --- rewrites that change the number of detecting regions naturally change the decoding problem, as they change the number of syndrome bits one gets. 
As such, restricting oneself to decoding-preserving rewrites is very limiting, to the extent that many rewrites in \cite{fault-tolerance-by-construction} are no longer available, making their proposed derivations impossible. 

In this work, we instead propose preserving certain desirable \emph{properties} about the decoding problem as a whole.
While this may enable substantially more rewrites, such rewrites are only applicable to contexts in which the relevant property of the decoding problem was present in the first place. 
Therefore, while this line of thinking enables more rewrites, it is only applicable in certain contexts. 
As such, it is essential to identify properties that satisfy three criteria.  
(1) They should be \emph{useful} for solving the decoding problem more efficiently. (2) They should be \emph{widely used} to make them applicable in many contexts and (3) they should be sufficiently \emph{easily preserved} to admit new and expressive rewrites.

Finding properties that satisfy all three of these criteria is difficult. 
In this work, we propose considering \emph{matchability}. 
Matchability requires that each edge in the ZX diagram is covered by at most two detecting regions \cite{pymatching-v1}. 
This is a well-studied property that admits efficient and good decoders by leveraging fast minimum-weight perfect matching algorithms \cite{pymatching-v1,pymatching-v2-sparse-blossom,fusion-blossom,micro-blossom}.
It is also a relatively widely-used property, as the surface code satisfies matchability. 
The surface code is one of the key quantum error correction codes studied by the likes of Google, in their recent, state-of-the-art experiments \cite{acharyaQuantumErrorCorrection2025}.
After defining \emph{matchability-preserving rewrites}, we will show that matchability is also relatively easily preservable, providing a set of rewrites to decompose arbitrary weight spiders. 
We will use these to show that any ZX diagram can be fault-equivalently synthesised into a quantum circuit while preserving matchability.

\section{Preliminaries}
\label{sec:bg}

The ZX calculus \cite{zx-calc-original,vandeweteringZXcalculusWorkingQuantum2020} is a diagrammatic reasoning tool for quantum programs.
Similar to the quantum circuit model, it gives a visual representation of the
underlying linear operations.
ZX diagrams are made up of only two basic building blocks:
\begin{definition}[Spider]
    Spiders are green (a.k.a.\ Z) or red (a.k.a.\ X) nodes with an angle
    $\alpha \in [0, \frac{\pi}{2}, \pi, \frac{3\pi}{2})$, a number of input legs $m \in \mathbb N_0$ and a
    number of output legs $n \in \mathbb N_0$.
    \begin{equation}
        \begin{gathered}
            \input{./figures/background/z-spider.tikz}\\
            := \ket 0^{\otimes n} \bra 0^{\otimes m} + e^{i\alpha} \ket 1^{\otimes n} \bra 1^{\otimes m}
        \end{gathered}
        \hspace{48pt}
        \begin{gathered}
            \input{./figures/background/x-spider.tikz}\\
            := \ket +^{\otimes n} \bra +^{\otimes m} + e^{i\alpha} \ket -^{\otimes n} \bra -^{\otimes m}
        \end{gathered}
    \end{equation}
\end{definition}
We can compose ZX diagrams sequentially, corresponding to matrix multiplication and in parallel, corresponding to the tensor product. 
Many common building blocks in the quantum circuit model can be compactly
represented as ZX diagrams (see \cref{fig:zx-representations}).
We remark that, by restricting the phases of spiders to be multiples of $\frac{\pi}{2}$ we restrict ourselves to Clifford ZX diagrams, which are universal for stabiliser states \cite{zx-calc-is-complete-for-stabilizer-quantum-mechanics}. 
For more details on the entire ZX calculus, see \cite{vandeweteringZXcalculusWorkingQuantum2020,picturing-quantum-software}.

\begin{figure}[]
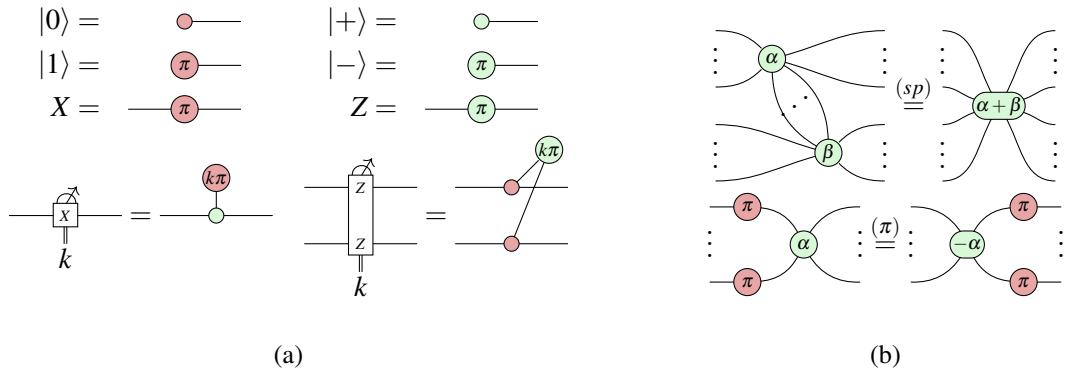

    \centering
    \begin{subfigure}[b]{0.49\textwidth}
        \centering
        \begin{equation*}
            \begin{gathered}
                \begin{aligned}
                    \ket0 &=& \input{./figures/background/zero-spider.tikz}
                        &\qquad& \ket+ &=& \input{./figures/background/plus-spider.tikz}\\
                    \ket1 &=& \input{./figures/background/one-spider.tikz}
                        &\qquad& \ket- &=& \input{./figures/background/minus-spider.tikz}\\
                    X &=& \input{./figures/background/x-flip.tikz}
                        &\qquad& Z &=& \input{./figures/background/z-flip.tikz}
                \end{aligned}
            %         \tikzfig{background/cnot-circuit}
            %         =
            %         \tikzfig{background/cnot-diagram}
            %         \hspace{24pt}
            %         \tikzfig{background/z-meas-circuit}
            %         =
            %         \tikzfig{background/z-meas-diagram}
            %         \hspace{24pt}
                \\
                \input{./figures/background/x-meas-circuit.tikz}
                =
                \input{./figures/background/x-meas-diagram.tikz}
                \quad
                \input{./figures/background/zz-meas-circuit.tikz}
                =
                \input{./figures/background/zz-meas-diagram.tikz}
            \end{gathered}
        \end{equation*}
        \caption{}
        \label{fig:zx-representations}
    \end{subfigure}
    \begin{subfigure}[b]{0.49\textwidth}
        \centering
        \begin{equation*}
            \begin{gathered}
                \input{./figures/background/sp-rule-LHS.tikz} \stackrel{(sp)}= \input{./figures/background/sp-rule-RHS.tikz}\\
                \input{./figures/background/pi-rule-LHS.tikz} \stackrel{(\pi)}= \input{./figures/background/pi-rule-RHS.tikz}
            \end{gathered}
        \end{equation*}
        \caption{}
        \label{fig:zx-rewrites}
    \end{subfigure}
    \caption{This figure demonstrates two core components in the ZX calculus.
    (a) shows how a few common building blocks from the quantum circuit model
    can be represented with few spiders in the ZX calculus.
    (b) summarises some important ZX rewriting rules.}
    \label{fig:bg:zx-overview}
\end{figure}

The key advantage of the ZX calculus comes from its powerful rewriting system.
The calculus allows a number of basic modifications to diagrams that keep the
underlying linear map intact.
\begin{definition}[ZX Rewrite]
    If $D_1$, $D_2$ are equivalent ZX diagrams, we call $D_1 = D_2$ a ZX
    rewrite.
\end{definition}
We can apply these rewrites to larger diagrams:
Replacing a sub-diagram with an equivalent one produces an equation between the original diagram and the one where the sub-diagram was replaced:
\begin{equation}
    \begin{tikzpicture}
	\begin{pgfonlayer}{nodelayer}
		\node [style=none] (0) at (4.5, 1) {};
		\node [style=none] (1) at (4.5, -1) {};
		\node [style=none] (2) at (5.5, -1) {};
		\node [style=none] (3) at (5.5, 1) {};
		\node [style=none] (4) at (4, 1.5) {};
		\node [style=none] (5) at (6, 1.5) {};
		\node [style=none] (6) at (7, 2.5) {};
		\node [style=none] (7) at (3, 2.5) {};
		\node [style=none] (8) at (4, -1.5) {};
		\node [style=none] (9) at (3, -2.5) {};
		\node [style=none] (10) at (6, -1.5) {};
		\node [style=none] (11) at (7, -2.5) {};
		\node [style=none] (12) at (4.5, 0.5) {};
		\node [style=none] (13) at (4.5, -0.5) {};
		\node [style=none] (14) at (5.5, 0.5) {};
		\node [style=none] (15) at (5.5, -0.5) {};
		\node [style=none] (16) at (6, 0.5) {};
		\node [style=none] (17) at (6, -0.5) {};
		\node [style=none] (18) at (4, 0.5) {};
		\node [style=none] (19) at (4, -0.5) {};
		\node [style=none] (20) at (7, 0.5) {};
		\node [style=none] (21) at (7, -0.5) {};
		\node [style=none] (22) at (8, 0.5) {};
		\node [style=none] (23) at (8, -0.5) {};
		\node [style=none] (24) at (2, 0.5) {};
		\node [style=none] (25) at (2, -0.5) {};
		\node [style=none] (26) at (3, 0.5) {};
		\node [style=none] (27) at (3, -0.5) {};
		\node [style=none] (28) at (5, 2) {$C$};
		\node [style=none] (29) at (5, 0) {};
		\node [style=none] (30) at (5, 0) {$D$};
		\node [style=none] (31) at (-9, 1) {};
		\node [style=none] (32) at (-9, -1) {};
		\node [style=none] (33) at (-8, -1) {};
		\node [style=none] (34) at (-8, 1) {};
		\node [style=none] (35) at (-9, 0.5) {};
		\node [style=none] (36) at (-9, -0.5) {};
		\node [style=none] (37) at (-8, 0.5) {};
		\node [style=none] (38) at (-8, -0.5) {};
		\node [style=none] (39) at (-7, 0.5) {};
		\node [style=none] (40) at (-7, -0.5) {};
		\node [style=none] (41) at (-10, 0.5) {};
		\node [style=none] (42) at (-10, -0.5) {};
		\node [style=none] (43) at (-8.5, 0) {$D$};
		\node [style=none] (44) at (-4, 1) {};
		\node [style=none] (45) at (-4, -1) {};
		\node [style=none] (46) at (-3, -1) {};
		\node [style=none] (47) at (-3, 1) {};
		\node [style=none] (48) at (-4, 0.5) {};
		\node [style=none] (49) at (-4, -0.5) {};
		\node [style=none] (50) at (-3, 0.5) {};
		\node [style=none] (51) at (-3, -0.5) {};
		\node [style=none] (52) at (-2, 0.5) {};
		\node [style=none] (53) at (-2, -0.5) {};
		\node [style=none] (54) at (-5, 0.5) {};
		\node [style=none] (55) at (-5, -0.5) {};
		\node [style=none] (56) at (-3.5, 0) {$E$};
		\node [style=none] (57) at (-6, 0) {=};
		\node [style=none] (58) at (0, 0) {$\Longrightarrow$};
		\node [style=none] (59) at (12.5, 1) {};
		\node [style=none] (60) at (12.5, -1) {};
		\node [style=none] (61) at (13.5, -1) {};
		\node [style=none] (62) at (13.5, 1) {};
		\node [style=none] (63) at (12, 1.5) {};
		\node [style=none] (64) at (14, 1.5) {};
		\node [style=none] (65) at (15, 2.5) {};
		\node [style=none] (66) at (11, 2.5) {};
		\node [style=none] (67) at (12, -1.5) {};
		\node [style=none] (68) at (11, -2.5) {};
		\node [style=none] (69) at (14, -1.5) {};
		\node [style=none] (70) at (15, -2.5) {};
		\node [style=none] (71) at (12.5, 0.5) {};
		\node [style=none] (72) at (12.5, -0.5) {};
		\node [style=none] (73) at (13.5, 0.5) {};
		\node [style=none] (74) at (13.5, -0.5) {};
		\node [style=none] (75) at (14, 0.5) {};
		\node [style=none] (76) at (14, -0.5) {};
		\node [style=none] (77) at (12, 0.5) {};
		\node [style=none] (78) at (12, -0.5) {};
		\node [style=none] (79) at (15, 0.5) {};
		\node [style=none] (80) at (15, -0.5) {};
		\node [style=none] (81) at (16, 0.5) {};
		\node [style=none] (82) at (16, -0.5) {};
		\node [style=none] (83) at (10, 0.5) {};
		\node [style=none] (84) at (10, -0.5) {};
		\node [style=none] (85) at (11, 0.5) {};
		\node [style=none] (86) at (11, -0.5) {};
		\node [style=none] (87) at (13, 2) {$C$};
		\node [style=none] (88) at (13, 0) {};
		\node [style=none] (89) at (13, 0) {$E$};
		\node [style=none] (90) at (9, 0) {=};
	\end{pgfonlayer}
	\begin{pgfonlayer}{edgelayer}
		\draw (8.center)
			 to (10.center)
			 to (5.center)
			 to (4.center)
			 to cycle;
		\draw (9.center)
			 to (7.center)
			 to (6.center)
			 to (11.center)
			 to cycle;
		\draw (0.center) to (3.center);
		\draw (3.center) to (2.center);
		\draw (2.center) to (1.center);
		\draw (1.center) to (0.center);
		\draw (18.center) to (12.center);
		\draw (19.center) to (13.center);
		\draw (14.center) to (16.center);
		\draw (15.center) to (17.center);
		\draw (20.center) to (22.center);
		\draw (21.center) to (23.center);
		\draw (24.center) to (26.center);
		\draw (25.center) to (27.center);
		\draw (31.center) to (34.center);
		\draw (34.center) to (33.center);
		\draw (33.center) to (32.center);
		\draw (32.center) to (31.center);
		\draw (41.center) to (35.center);
		\draw (42.center) to (36.center);
		\draw (37.center) to (39.center);
		\draw (38.center) to (40.center);
		\draw (44.center) to (47.center);
		\draw (47.center) to (46.center);
		\draw (46.center) to (45.center);
		\draw (45.center) to (44.center);
		\draw (54.center) to (48.center);
		\draw (55.center) to (49.center);
		\draw (50.center) to (52.center);
		\draw (51.center) to (53.center);
		\draw (67.center)
			 to (69.center)
			 to (64.center)
			 to (63.center)
			 to cycle;
		\draw (68.center)
			 to (66.center)
			 to (65.center)
			 to (70.center)
			 to cycle;
		\draw (59.center) to (62.center);
		\draw (62.center) to (61.center);
		\draw (61.center) to (60.center);
		\draw (60.center) to (59.center);
		\draw (77.center) to (71.center);
		\draw (78.center) to (72.center);
		\draw (73.center) to (75.center);
		\draw (74.center) to (76.center);
		\draw (79.center) to (81.center);
		\draw (80.center) to (82.center);
		\draw (83.center) to (85.center);
		\draw (84.center) to (86.center);
	\end{pgfonlayer}
\end{tikzpicture}

\end{equation}
In what follows, we will use the following notation: $C(D)$ denotes a diagram with a subdiagram $D$, and the diagram $C(D)$ with the subdiagram $D$ \emph{removed} is denoted $C$. 

A fundamental rewriting rule is Only Connectivity Matters (OCM), meaning that we can deform ZX diagrams arbitrarily as long as the spiders and their connectivity remain the same, much like graph isomorphisms.
Furthermore, there is a complete set of rewrites, meaning that any two ZX
diagrams that represent the same quantum program can be rewritten into each
other through a finite sequence of rewrite rule applications.
We present a few ZX rewrites that we will frequently use in this work in
\cref{fig:zx-rewrites}.

More recently, the ZX calculus has found wide application in the domain of quantum error correction (QEC) \cite{unifying-flavors-of-ft-with-the-zx-calc, fault-tolerance-by-construction,floquetifying-the-colour-code, scalable-spider-nests, Gidney2023}. 
A key observation when reasoning about noisy circuits is that two circuits, even if they implement the same linear map, may behave very differently under noise, meaning that faults on one circuit may propagate substantially worse than faults on another circuit. 
Therefore, \cite{fault-tolerance-by-construction} propose restricting the set of allowed rewrites to the ones that not only preserve the underlying computation but also the circuit's behaviour under noise, calling such rewrites \emph{fault-equivalent}.
When manipulating circuits in the noisy setting, it is sensible to restrict oneself to such rewrites. 
We refer to \cite{fault-tolerance-by-construction} for more details on fault equivalence. 

Another useful tool in the context of QEC are \emph{Pauli webs} \cite{unifying-flavors-of-ft-with-the-zx-calc}.
Pauli webs are a decoration of ZX diagrams, where wires are highlighted in green and/or red to encode stabilisers of the underlying operation.
\begin{definition}[Pauli web]
    For a ZX diagram $D$, a Pauli web $P : \rm{edges}(D) \to \{I, Z, X, ZX\}$, assigns colours to edges.
    For every spider with a phase $\alpha \in \left\{ 0, \pi \right\}$ in the
    diagram, $P$ covers an even number of the spider's legs in the colour of the
    spider and either all or none of the spider's legs in the opposite colour
    \footnote{While Pauli webs can be defined for the full Clifford fragment, we
    will only use them for phaseless diagrams in this work.}.
\end{definition}
A special kind of Pauli webs are those that are non-trivial but colour all
boundary edges (those adjacent to less than two spiders) trivially. 
We call these Pauli webs \emph{detecting regions} since they characterise the
parity checks of QEC spacetime codes \cite{delfosseSpacetimeCodes2023}.
Detecting regions can detect errors they anticommute with, meaning that the corresponding detecting set will give a non-trivial syndrome \cite{floquetifying-the-colour-code,floquetifying-stabiliser-codes}.

Pauli webs can be helpful for analysing \emph{matchability}, the property that
enables Minimum Weight Perfect Matching (MWPM) decoding.
MWPM decoding (and therefore matchability) is highly relevant to the practical
application of Quantum Error Correction because (1) it is scalable with a
polynomial time complexity and efficient implementations available
\cite{pymatching-v1,pymatching-v2-sparse-blossom,fusion-blossom,micro-blossom}
and because (2) it is highly relevant as the surface code, one of the most
widely studied QEC codes, is matchable \cite{surface-code}.
From a technical perspective, \emph{matchability} requires each atomic fault to violate at most two detectors within some basis of all detectors \cite{pymatching-v1}.
Under the edge flip noise model \cite{floquetifying-stabiliser-codes}, we can
characterise matchability with the following diagrammatic definition.
\begin{definition}[Matchability]
    \label{def:matchability}
    Let $\mathcal S$ be a basis of detecting Pauli webs for a ZX diagram $D$.
    If for every edge $e \in \rm{edges}(D)$, there are at most two Pauli webs
    $P \in \mathcal S$ covering $e$ nontrivially ($P(e) \neq I$) then we call
    $\mathcal S$ a matchable detector basis and $D$ matchable with $\mathcal S$.
    For a ZX diagram where all detecting regions are single-coloured, $\mathcal S$ is \emph{CSS matchable} if every edge is covered by at most two Pauli webs in each colour.
\end{definition}
CSS matchable Pauli webs make use of the fact that when all detecting regions are single-coloured, we can split the decoding problem into two separate problems; one tackling $X$ errors and the other tackling $Z$ errors. 
\section{Matchability preservation}
\label{sec:mp}

ZX rewrites allow us to modify ZX diagrams without changing the underlying linear map. 
Restricting to fault-equivalent rewrites further ensures that the diagram's behaviour under noise is preserved. 
However, such rewrites can change the detecting regions of the diagram and may therefore change the associated decoding problem. 

We consider rewrites that specify not only their action on the diagram, but also their effect on the decoding problem. 
We call these \emph{detector-aware rewrites}. 
They extend standard ZX rewrites by tracking additional data about changes to the detecting regions.
So, a standard ZX rewrite can be part of many different detector-aware rewrites. 
We then define \emph{matchability-preserving rewrites}, which preserve the property that the underlying decoding problem is matchable. 
In the following sections, we will then show how any phase-free ZX diagram can be fault-equivalently extracted into a circuit while preserving matchability.

\subsection{Detector-aware rewrites}
\label{sec:mp:da-rewrites}
To define detector-aware rewrites, we first have to isolate the additional data that we need to track. 
In particular, we are interested in how rewrites change the detecting regions of a diagram. 
Therefore, we have to consider the restriction of Pauli webs to the subdiagram we want to rewrite.

We define:
\begin{definition}[Stabilising and detecting webs, local detector basis]
 Let $\mathcal{S}$ be a multiset of Pauli webs on $D$. 
    $\mathcal{S}|_S$ indicates the subset of $\mathcal{S}$ that are stabilisers of $D$, i.e.\@ have non-trivial action on the boundary, and $\mathcal{S}|_D$ indicates the subset of detecting regions, i.e.\@ have trivial action on the boundary.
 We say a Pauli web multiset $\mathcal{S}$ of a ZX diagram $D$ forms a \emph{local detector basis} if $\mathcal S|_D$ forms a basis for the detecting regions of $D$.
\end{definition}

Next, we define:
\begin{definition}[Local restriction]
 Let $C$ be a ZX diagram and let $D$ be a sub-diagram of $C$.
 For a Pauli web $P$ on $C$, we define the \emph{local restriction} of $P$ to
    $D$, written as $P[D]$, to be the action of $P$ on the edges of $D$.
 We say a local restriction is \emph{trivial} if it acts as identity on all edges in $D$.
 For a multiset $\mathcal S$ of Pauli webs on $C$, the local restriction $\mathcal S[D]$ of $\mathcal S$ to $D$ is the multiset of all non-trivial, local restrictions of Pauli webs from $\mathcal S$ to $D$.
\end{definition}
Note that we use multisets of Pauli webs here rather than regular sets.
This is because different Pauli webs on the full diagram can have the same local
restriction on some sub-diagrams:

\begin{example}
    Consider the following diagram $C(D)$ ($D$ is highlighted by the dashed box)
    with four detecting regions $\mathcal{S} = \{A, B, C, D\}$:
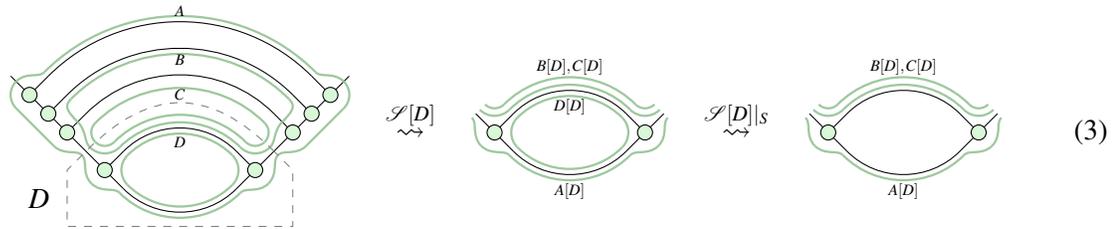
    \begin{equation}
        \input{./figures/matchability-preservation/da-rewrites/restriction-global.tikz}
        \quad\stackrel{\mathcal S[D]}\rightsquigarrow\quad
        \input{./figures/matchability-preservation/da-rewrites/restriction-local.tikz}
        \quad\stackrel{\mathcal S[D]|_S}\rightsquigarrow\quad
        \input{./figures/matchability-preservation/da-rewrites/restriction-local-stabilisers.tikz}
    \end{equation}
    As $B[D] = C[D]$, $\mathcal{S}[D] = \{A[D], 2 \times B[D], D[D]\}$ (using
    ``$2\times$'' to indicate the multiplicity of $B[D]$).
    Regions $A, B$ and $C$ become stabilising Pauli webs on $D$ with the same
    boundary action.
    Only $D$ remains a detecting region on $D$, so
    $\mathcal{S}[D]|_S = \{A[D], 2 \times B[D]^2\}$.
\end{example}

To specify how decoder-aware rewrites update multisets of Pauli webs, we first have to define when they are compatible. 
For some multiset $\mathcal{S}$, we define the expanded set of $\mathcal{S}$ as $\mathcal{S}^* := \{(S, i) | S \in \mathcal{S}, i = 1, ..., \text{multiplicity of } S \text{ in } \mathcal{S}\}$. 
We say:
\begin{definition}[Boundary-respecting coupling]
 Let $D, E$ be equivalent ZX diagrams with boundary edges $B_D$ and
    $B_E$. 
 Let $\mathcal{S}_D$, $\mathcal{S}_E$ be Pauli web multisets on $D$ and $E$ respectively.
 A boundary coupling between $\mathcal{S}_D$, $\mathcal{S}_E$ is a bijection $R: (\mathcal{S}_D|_S)^* \to (\mathcal{S}_E|_S)^*$ between the Pauli webs with non-trivial action on the boundary. 
 We say $R$ is boundary-respecting if for all $S \in \mathcal{S}_D|_S$, we have $S[B_D] = R(S)[B_E]$.
\end{definition}

A boundary coupling maps Pauli webs with non-trivial action on the boundary to other Pauli webs with non-trivial action on the boundary. 
It is boundary-respecting if two webs that are mapped to each other have the same effect on the boundary. 

\begin{example}
    Returning to our example above, we have
    $\mathcal{S}[D]|_S = \{A[D], 2 \times B[D]\}$ and thus,
    $(\mathcal{S}|_S)^* = \{(A[D], 1), (B[D], 1), (B[D], 2)\}$.
    For example, we could consider the following rewrite:
    \begin{equation}
        \input{./figures/matchability-preservation/da-rewrites/same-boundary-coverage-LHS.tikz}
        \quad\rightsquigarrow\quad
        \input{./figures/matchability-preservation/da-rewrites/same-boundary-coverage-RHS.tikz}
    \end{equation}
    Here, the RHS has two Pauli webs going via the bottom edge and only one
    going via the top edge. 
    One boundary-respecting coupling would be:
    $((B[D], 1), (A', 1)), ((B[D], 2), (C', 1)), ((A[D], 1), (C', 2))$. 
\end{example}

Finally, we can define: 
\begin{definition}[Detector-aware rewrite]
 Let $D \FaultEq E$ be a fault-equivalent ZX rewrite. Then a detector-aware rewrite consists of 
    \begin{itemize}
        \item a local detector basis $\mathcal{S}_D$ for $D$
        \item a local detector basis $\mathcal{S}_E$ for $E$
        \item a boundary-respecting coupling $R: (\mathcal{S}_D|_S)^* \to (\mathcal{S}_E|_S)^*$
    \end{itemize}
 We use $(D, \mathcal{S}_D) \underset{R}{\circeq} (E, \mathcal{S}_E)$ to denote a detector-aware rewrite.
    % When the diagrams are clear from context, we write $S_1 \circeq S_2$. 
\end{definition}
Note that this definition requires fault equivalence.
While it is conceivable to think about detector-aware rewrites that are not fault-equivalent, in most reasonable contexts, when thinking about decoding, we also want to preserve the circuit's behaviour under noise.

Finally, we can define how to apply detector-aware rewrites: 
\begin{definition}[Applying detector-aware rewrites]
    \label{thm:applying-da-rewrites}
 Let $C(D)$ be some ZX diagram with a subdiagram $D$ and some Pauli web multiset $\mathcal{S}$. 
 Let $(D, \mathcal{S}[D]) \underset{R}{\circeq} (E, \mathcal{S}_E)$ be a detector-aware rewrite.
 Then applying $(D, \mathcal{S}[D]) \underset{R}{\circeq} (E, \mathcal{S}_E)$ to $C(D)$ gives a new diagram $C(E)$, along with a new Pauli web multiset $\mathcal{S}'$ consisting of: 
    \begin{itemize}
        \item For each detecting region $S \in \mathcal{S}_E|_D$, a web that acts trivially on $C$ and as $S$ on $E$.
        \item For each Pauli web $S \in \mathcal{S}$ that acts non-trivially on $C$ and trivially on the boundary of $D$, a web that acts as $S$ on $C$ and trivially on $E$.
        \item For each Pauli web $S \in \mathcal{S}$ that acts non-trivially on the boundary of $D$, a web that acts as $S$ on $C$ and as $R(S[D])$ on $E$.
    \end{itemize}
\end{definition}
So, $\mathcal{S}$ is updated by replacing the local detecting regions in $D$ with those in $E$ (i.e.\@ the webs that act trivially on $C$) and then updating the remaining webs according to $R$. 
\begin{example}
    Returning to our example above, applying the detector-aware rewrite
    specified above, we get:
    \begin{equation}
        \input{./figures/matchability-preservation/da-rewrites/rewrite-LHS.tikz}
        \quad\underset R \circeq\quad
        \input{./figures/matchability-preservation/da-rewrites/rewrite-RHS.tikz}
    \end{equation}
\end{example}

In \cref{appendix:properties-of-da-rewrites}, we show that detector-aware rewrites can be applied transitively, meaning that we can chain a sequence of detector-aware rewrites and still get a new detector-aware rewrite. 
Furthermore, we show that they are compositional, meaning that if we have a detector-aware rewrite from $D$ to $E$, we can naturally define a corresponding rewrite from $C(D)$ to $C(E)$. 
Finally, we show that if we have a detector basis on some diagram, detector-aware rewrites give a new set of Pauli webs that form a detector basis. 
In other words, we show that detector-aware rewrites behave as expected of local rewrites. 

\subsection{Matchability-preserving rewrites}
\label{sec:mp:mp-rewrites}
We have defined detector-aware rewrites as ZX rewrites that track additional information about how the decoding problem should be updated under the rewrite. 
So far, we have put no restrictions on which updates to the decoding problem are considered valid. 
\cite{fault-tolerant-transformations-of-spacetime-codes} propose to restrict to rewrites that preserve the decoding problem. 
However, any rewrite that changes the number of detecting regions would naturally violate that rather restrictive constraint. 
Therefore, instead, we will propose an alternative property to preserve: \emph{matchability}.
We recall that matchability was defined such that every edge is covered by at most two detecting regions. 
Furthermore, CSS matchability requires that each edge is covered by at most two detecting regions of each colour (see \cref{def:matchability}). 

We define:
\begin{definition}[Matchability-preserving rewrite]
    \label{def:mp-rewrite}
 A detector-aware rewrite $(D, \mathcal{S}) \underset{R}{\circeq} (E, \mathcal{T})$ is called a \emph{matchability-preserving
 rewrite} if  $\mathcal{S}$ and $\mathcal{T}$ are matchable Pauli web multisets.
 We say that it is \emph{CSS matchability-preserving} if $\mathcal{S}$ and $\mathcal{T}$ are CSS matchable Pauli web multisets.
\end{definition}

We can see that matchability-preserving rewrites are independent of the choice
of a boundary-respecting coupling. 
If a rewrite $(D, \mathcal{S}) \underset{R}{\circeq} (E, \mathcal{T})$ is
matchability-preserving, then naturally, any other rewrite
$(D, \mathcal{S}) \underset{R'}{\circeq} (E, \mathcal{T})$ with a different
coupling choice is also matchability-preserving.
We will therefore often omit the choice of the coupling in the remainder of this
paper, just showing that at least one coupling exists. 

We can now state:
\begin{proposition}[Matchability-preserving rewrites preserve matchability]
    Let $C(D)$ be a ZX diagram with a subdiagram $D$ and let $\mathcal S$
    be a \emph{matchable} Pauli web multiset of $C(D)$. 
    Then applying a matchability-preserving rewrite
    $(D, \mathcal S[D]) \underset{R}{\circeq} (E, \mathcal T)$ to $C(D)$ gives
    $C(E)$ along with a new Pauli web multiset $\mathcal S'$ that is matchable. 
\end{proposition}
\begin{proof}
    By \cref{thm:applying-da-rewrites}, we know that the only detecting
    regions that highlight edges in $C$ are detecting regions that are also
    $\mathcal S$.
    (Recall that $C$ denotes $C(D)$ with $D$ removed.)
    As we assumed that $\mathcal S$ is matchable, this means that each edge in
    $C$ must also satisfy the matchability criteria in $\mathcal S'$.
    Furthermore, edges can only be acted upon by Pauli webs specified in
    $\mathcal T$. 
    By definition, $\mathcal{T}$ is matchable and thus $\mathcal S'$ must be
    matchable on $C(E)$.
\end{proof}

We can now give some matchability-preserving rewrites: 
\begin{proposition}[Id-removal]
 The following rewrite is (CSS) matchability preserving:
    \begin{equation}
        \begin{tikzpicture}
	\begin{pgfonlayer}{nodelayer}
		\node [style=none] (0) at (-1.5, 0) {};
		\node [style=none] (1) at (1.5, 0) {};
	\end{pgfonlayer}
	\begin{pgfonlayer}{edgelayer}
		\draw [style=ZZ W XX] (0.center) to (1.center);
	\end{pgfonlayer}
\end{tikzpicture}

        \quad\stackrel{(r_{elim})}\circeq\quad
        \input{./figures/matchability-preservation/mp-rewrites/id-RHS.tikz}
    \end{equation}
\end{proposition}
\begin{proof}
    \cite[Prop.~6.12]{fault-tolerance-by-construction} show that the rewrite is
 fault-equivalent. 
 As neither side has detecting regions, the provided Pauli web multisets
 trivially form a detecting region basis. 
 Furthermore, there exists an obvious boundary-respecting coupling between the Pauli webs. 
\end{proof}

\begin{proposition}[Spider-Unfuse]
 For any (CSS) matchable Pauli web multiset on either side of this rewrite, there
 exists a (CSS) matchable Pauli web multiset on the other side, such that this
 rewrite is matchability preserving:
    \begin{equation}
        \input{./figures/matchability-preservation/mp-rewrites/unfuse-LHS.tikz}
        \quad\stackrel{(r_{fuse})}\circeq\quad
        \input{./figures/matchability-preservation/mp-rewrites/unfuse-RHS.tikz}
    \end{equation}
\end{proposition}
\begin{proof}
    \cite[Prop.~6.13]{fault-tolerance-by-construction} show that the rewrite is
 fault-equivalent. 
 As neither side has detecting regions, any matchable Pauli web multiset will 
 trivially form a detecting region basis. 
 For any (CSS) matchable multiset of Pauli webs on the left diagram, we can construct a matchable multiset of Pauli webs on the right by, for each Pauli web, taking the same highlighting on the boundary edges, and if the boundary edges are highlighted in red, also highlighting the internal edge in red. 
 This must give valid Pauli webs on the right, along with an obvious boundary-respecting coupling. 
 As the original multiset was (CSS) matchable, the new multiset must also be (CSS) matchable.

 Similarly, for any (CSS) matchable multiset of Pauli webs on the right, we can construct a (CSS) matchable multiset of Pauli webs on the left by, for each Pauli web, dropping the action on the internal edge. 
 All Pauli webs in the new set must still be valid, as the internal edge could not have been highlighted in green, and if it was highlighted in red, then all other edges must also be highlighted in red. 
 Once again, this gives an obvious boundary-respecting coupling, and as the original multiset was (CSS) matchable, the new multiset must also be (CSS) matchable.
\end{proof}

\begin{proposition}[$r_4$]
    The following rewrite is CSS matchability preserving:
    \begin{equation}
        \input{./figures/matchability-preservation/mp-rewrites/r4-blossom-LHS.tikz} \quad \cup \quad 2 \times \input{./figures/matchability-preservation/mp-rewrites/r4-blossom-LHS-red.tikz} \qquad 
        \stackrel{(r_4)}\circeq\qquad
        \input{./figures/matchability-preservation/mp-rewrites/r4-blossom-RHS.tikz} \quad \cup \quad 2 \times \input{./figures/matchability-preservation/mp-rewrites/r4-blossom-RHS-red.tikz}
    \end{equation}
    \textbf{Remark:} 
    The rewrite above contains six Pauli webs with non-trivial boundary action
    (four green, two red).
    The union of diagrams refers to the union of the corresponding Pauli web
    multisets.
    Webs highlighting more than two edges of one spider cannot be drawn as a
    single line, so we instead highlight edges, to the usual presentation of
    Pauli webs, and use ``$2 \times$'' to indicate their multiplicity. 
\end{proposition}
\begin{proof}
 Rodatz et al. prove that $r_4$ is a fault-equivalent rewrite in
    \cite[Prop.~6.14]{fault-tolerance-by-construction}.
 The left diagram has no detecting regions, and the right diagram contains a Pauli web that generates the single detecting region present. 
 Even though not explicitly stated, there exists an obvious boundary-respecting coupling that maps each of the webs with non-trivial action on the boundary to their corresponding web on the other side.  
 So, the rewrite is a valid detector-aware rewrite. 
 Furthermore, it is matchable since every wire is covered by at most two Pauli webs of each colour.
\end{proof}

\section{Matchability-preserving circuit extraction}
\label{sec:mp:generality}

Given a phase-free ZX diagram, we can do fault-tolerant circuit extraction using the following two steps \cite{floquetifying-stabiliser-codes}\footnote{\cite{fault-tolerance-by-construction} allow idealised edges on ZX diagrams that are treated as fault-free. A general circuit-extraction procedure would furthermore have to consider how to unidealise edges. This is, however, as of yet, an open problem and therefore not further considered in this work. Instead, we focus on fault-equivalent circuit extraction from ZX diagrams without idealised edges.}: 
(1) decompose  all spiders into spiders of weight three, (2) use ($r_{elim}$) and ($r_{fuse}$) to extract a circuit.
As we have already shown that ($r_{elim}$) and ($r_{fuse}$) are matchability-preserving, it remains to be shown that we can decompose arbitrary weight spiders into spiders of weight three. 
While there is significant room for optimising such protocols, in this work we focus on showing that it is possible to carry out this procedure possible while preserving matchability. 

When considering arbitrary weight spiders in arbitrary, matchable contexts, we must define multiple matchability-preserving rewrites: one for each possible set of Pauli webs. 
However, we can show that if we have a matchability-preserving rewrite, dropping some of the Pauli webs results in a new matchability-preserving rewrite: 
\begin{proposition}[Submultisets]
	\label{prop:submultisets}
    Let $(D, \mathcal{S}) \underset{R}{\circeq} (E, \mathcal{T})$ be a matchability-preserving rewrite. 
    Then for any $\mathcal{S'} \subseteq \mathcal{S}$ such that $\mathcal{S'}$ forms a detector basis for $D$, $(D_1, \mathcal{S'}) \underset{R|_{\mathcal{S'}}}{\circeq} (E, \{R(S) | S \in \mathcal{S}|_B\} \cup \mathcal{T}|_D)$ is a matchability-preserving rewrite. 
\end{proposition}
\begin{proof}
    As $(D, \mathcal{S}) \underset{R}{\circeq} (E, \mathcal{T})$, we must have $D \FaultEq E$. 
    Furthermore, by definition, $\mathcal{S'}$ and $\{R(S) | S \in \mathcal{S}|_B\} \cup \mathcal{T}|_D$ form detector bases of $D$ and $E$ respectively and $R|_{\mathcal{S'}}$ provides a boundary-respecting matching.  
    As $\mathcal{S}$ and $\mathcal{T}$ are matchable Pauli webs, any subsets of them must be matchable. 
\end{proof}
As a consequence of \cref{prop:submultisets}, if we can identify a pattern of detecting regions and show that nearly all possible Pauli web configurations are subsets of that pattern, we will have provided all necessary matchability-preserving rewrites.

First, we show: 
\begin{proposition}[$r_{4n}$ preserves matchability]
    \label{prop:floral-decompose}
    The following rewrites preserve CSS matchability:
    \begin{equation}
        \input{./figures/circuit-extraction/r4n-blossom-LHS.tikz}
        \hspace{4pt}\cup\hspace{4pt}
        2 \times \input{./figures/circuit-extraction/r4n-blossom-LHS-red.tikz} 
        \quad \stackrel{(r_{4n})}\circeq \quad 
        \input{./figures/circuit-extraction/r4n-blossom-RHS.tikz}
        \hspace{4pt}\cup\hspace{4pt}
        2 \times \input{./figures/circuit-extraction/r4n-blossom-RHS-red.tikz} 
    \end{equation}
    Here, the dots on the internal wires refer to bundles of $n$ wires with
    local detectors in between them.
\end{proposition}
\begin{proof}
    First, we need to show fault equivalence between the two diagrams. 
    The proof builds on the observation that for any undetectable, internal $X$ fault between two spiders to occur, all $n$ of the connecting edges have to have an $X$ fault. 
    Based on this, the proof is analogous to the $r_4$ fault equivalence proof in
    \cite[Prop.~A.4]{fault-tolerance-by-construction}.

    Next, we can observe that the Pauli web multisets form local detector bases on both sides, and that there exists an obvious boundary-respecting coupling between the multisets. 
    Therefore, this is a detector-aware rewrite. 
    Furthermore, as the Pauli web multisets highlight each edge at most twice in each colour, the rewrite is matchability-preserving. 
\end{proof}

While this matchability-preserving rewrite covers several cases, not all matchable detecting regions around a single spider will have the flower-shaped pattern above; 
all that we are guaranteed is that each Pauli web covers an even number of legs of the spider in question. 
If some Pauli web enters on edge $e_1$, it will leave on some edge $e_2$. 
There might be another Pauli web that enters on $e_2$ and leaves on $e_3$. 
However, the next Pauli web, entering on $e_3$, might leave on edge $e_1$ rather than $e_4$, breaking the flower-shaped pattern. 
To handle such cases, we provide a rewrite that decomposes large spiders with multiple such smaller flower patterns into individual spiders with flower patterns. 
For this rewrite to be valid, we require the \emph{distance} of the overall circuit to be sufficiently large. 
The circuit distance quantifies how good a quantum circuit is at detecting and correcting errors. 
It is defined as the minimum weight of all non-trivial, undetectable faults: if a circuit has distance $d$, all faults of weight less than $d$ must either be trivial or detectable. 

\begin{proposition}
    \label{rewrite:unfuse-floral}
    Let $C(D)$ have a circuit distance of more than $n$. Then, the following rewrite is matchability-preserving:
    \begin{equation}
        \input{./figures/circuit-extraction/double-detector-dotted-LHS.tikz} \quad \cup \quad 2 \times 
        \input{./figures/circuit-extraction/double-detector-dotted-LHS-red.tikz} \qquad 
        \circeq \qquad
        \input{./figures/circuit-extraction/double-detector-dotted-RHS.tikz} \quad \cup \quad 2 \times 
        \input{./figures/circuit-extraction/double-detector-dotted-RHS-red.tikz}
    \end{equation}
    where we have a flower-shaped pattern on the $n$ wires on the left and matchable Pauli webs on the right. 
\end{proposition}
\begin{proof}
    First, we show that the rewrite is fault-equivalent. 
    For this, we observe that we only have to worry about faults on the newly unfused edges. 
    A $Z$-type fault propagates to the boundary without increasing in weight and is therefore fine. 
    An $X$-type fault propagates to all the inputs, creating a fault of weight $n$. 
    This fault is undetectable as it commutes with all detecting regions. 
    As we previously assumed the circuit distance to be more than $n$, this means that, on the original diagram, this undetectable fault of weight $n$ must be trivial. 
    Therefore, an $X$-type fault on the unfused edge must be trivial. 

    As neither diagram has detectors, the Pauli web multisets trivially form a basis for the detecting regions. 
    Thus, the rewrite is a valid detector-aware rewrite. 
    Furthermore, as the Pauli web multisets cover each edge at most twice, the rewrite is matchability-preserving. 
\end{proof}

This rewrite results in a spider that has a flower-shaped pattern on the left spider, except for one edge not being highlighted by a green detecting region at all. 
By \cref{prop:appx:merging-stabilisers}, this means we can apply \cref{prop:floral-decompose} to decompose it further.  

Finally, we can state the main theorem of this paper: 
\begin{theorem}[CSS matchability-preserving circuit synthesis]
    Let $D$ be a CSS matchable, phase-free ZX diagram with circuit distance $d$ where all spiders have degree at less than $d$. 
    Then we can CSS matchability-preservingly extract a fault-equivalent quantum circuit from $D$.
\end{theorem}
\begin{proof}
    All phase-free spiders can be decomposed as follows:
    \begin{enumerate}
        \item
            Use \cref{rewrite:unfuse-floral} rewrite to unfuse all spiders until
            their surrounding detecting regions form flower patterns.
        \item
            Use \cref{prop:floral-decompose} rewrite to decompose all spiders
            into spiders of weight at most four.
            If a spider has a degree that is not a multiple of four, we can use
            ($r_{fuse}$) to increase its degree.
        \item
            Use $r_4$ to decompose all spiders into spiders of weight three.
        \item
            Use ($r_{elim}$) and ($r_{fuse}$) to extract a circuit.
    \end{enumerate}
\end{proof}

\section{Worked example}
\label{sec:mp:example}
We demonstrate the methodology from the previous sections through an application to plaquette measurements in the rotated surface code.
More specifically, we use matchability-preserving rewrites to find a matchable, fault-equivalent implementation for the $Z^4$ stabiliser measurement, surrounded by plaquettes on all sides.
We remark that, in the context of the surface code, a single, bare ancilla measurement is fault-tolerant and matchability-preserving. However, alternative syndrome extraction procedures are sometimes necessary, for instance when compiling to different gate sets, as for example explored in e.g. \cite{Gidney2023, improved-pairwise-measurement-based-surface-code}.
Here, we will show a simple example of a matchability-preserving circuit extraction. 
In \cref{sec:appx:examples-continued}, we furthermore show that the construction of \cite{improved-pairwise-measurement-based-surface-code} is, in fact, matchability-preserving, contrary to what the paper states. 

Following \cite{floquetifying-stabiliser-codes,fault-tolerance-by-construction},
we represent the four-qubit Pauli measurement as a four-legged spider without a phase.
For the repeated measurement of a Z plaquette that is surrounded on all four sides by X plaquettes, we get the following Pauli web coverage.
Note that for the purpose of brevity, we focus on the red Pauli webs in this
section.
Tracking the green Pauli webs is straightforward in this example, and we provide
the full rewriting steps in \cref{sec:appx:examples-continued}.
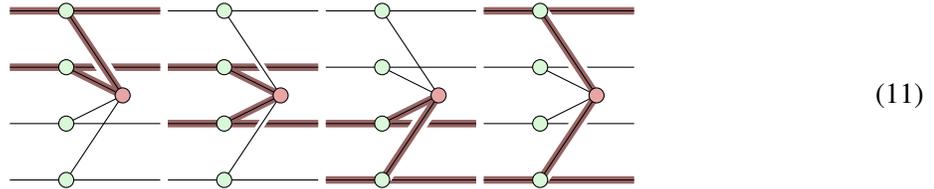
\begin{equation}
    \input{./figures/example/x-web-1.tikz}
    \input{./figures/example/x-web-2.tikz}
    \input{./figures/example/x-web-3.tikz}
    \input{./figures/example/x-web-4.tikz}
\end{equation}

Focusing on the four-legged spider in the middle, we apply
matchability-preserving decompositions until we get a quantum circuit:
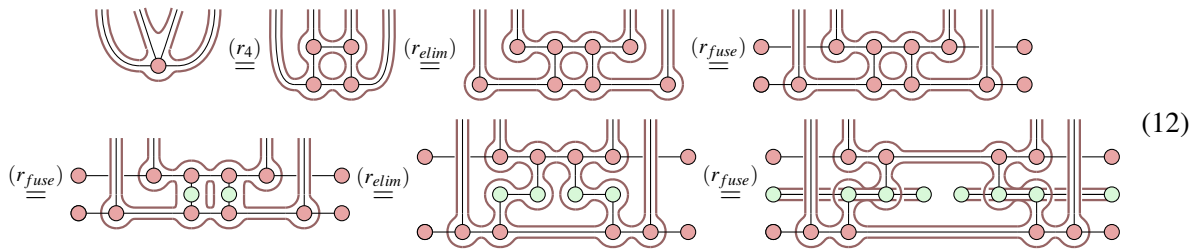
\begin{equation}
    \begin{gathered}
        \input{./figures/example/zx-style-1-x-coloured.tikz}
        \stackrel{(r_4)}= \input{./figures/example/zx-style-3-x-coloured.tikz}
        \stackrel{(r_{elim})}= \input{./figures/example/zx-style-4-x-coloured.tikz}
        \stackrel{(r_{fuse})}= \input{./figures/example/zx-style-5-x-coloured.tikz}\\
        \stackrel{(r_{fuse})}= \input{./figures/example/zx-style-6-x-coloured.tikz}
        \stackrel{(r_{elim})}= \input{./figures/example/zx-style-7-x-coloured.tikz}
        \stackrel{(r_{fuse})}= \input{./figures/example/zx-style-8-x-coloured.tikz}
    \end{gathered}
\end{equation}
Note that in the last rewrite of the $X$ detectors multiset, our notation becomes somewhat convoluted on the green spiders.
This is because the Pauli webs need to cover all three legs of these green spiders.
Reverting to standard Pauli web notation where we draw each web on an individual diagram, we get:
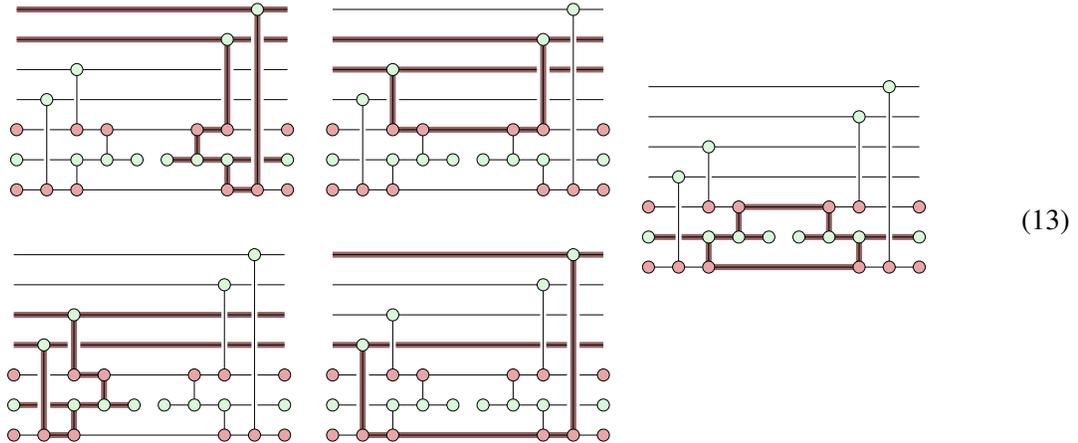
\begin{equation}
        \begin{array}{ccc}\scalebox{0.8}{
            \input{./figures/example/zx-style-9-x-web-1.tikz}}
                & \scalebox{0.8}{\input{./figures/example/zx-style-9-x-web-2.tikz}}
                & \multirow{2}{*}{\scalebox{0.8}{\input{./figures/example/zx-style-9-x-web-5.tikz}}}
            \\[50pt]
            \scalebox{0.8}{\input{./figures/example/zx-style-9-x-web-3.tikz}}
                & \scalebox{0.8}{\input{./figures/example/zx-style-9-x-web-4.tikz}}
        \end{array}
\end{equation}

Another example of matchability-preserving rewrites in action can be found in
\cref{sec:appx:examples-continued}.
In this example, we rederive an implementation of a weight-four syndrome measurement
for the surface code using non-destructive weight-two measurements, proposed by
\cite{improved-pairwise-measurement-based-surface-code}.
In their paper, the authors claim that their new circuit does not have a
matchable basis of detecting regions and, therefore, propose a new splitting-based decoder that requires syndrome-dependent reweighting of the edges in the decoding graph in order to ensure appropriate matchings. 
However, we show that their circuits can be derived using local,
matchability-preserving rewrites, thereby proving that a matchable basis must
exist and the authors could have used a regular splitting decoder.\footnote{We note that even with a (CSS) matchable basis, it is generally necessary to split some faults from the circuit noise model.}

\section{Conclusion}
\label{sec:conclusion}

In this work, we augmented ZX rewrites to track changes to the decoding problem
along with changes to ZX diagrams.
We applied these detector-aware rewrites to develop matchability-preserving
rewrites, which ensure that the ability to
apply efficient MWPM decoding in the resulting spacetime code is preserved.
Furthermore, we demonstrated a set of matchability-preserving rewrites that is
sufficient for synthesising MWPM-decodable, fault-tolerant implementations of phase-free ZX diagrams, including syndrome measurements in matchable codes. 

Future work includes generalising matchability-preserving circuit synthesis beyond the phase-free fragment of the ZX calculus. 
A first step would be to extend to the Clifford fragment, before extending to arbitrary ZX diagrams.
The latter requires generalising Pauli webs beyond Clifford ZX diagrams. 
Furthermore, alternative rewrite strategies could be explored, for example optimising for gate count or considering hardware constraints. 
In \cref{appendix:rewrite-variants}, we propose alternative matchability-preserving rewrites that could be helpful for this. 
Additionally, we would like to explore other desirable properties of the decoding problem that can be preserved under ZX rewrites.
For example, it would be interesting to consider rewrites that preserve the girth of the decoding graph, allowing for high-accuracy belief propagation decoding. 
Finally, one could consider rewrites that allow leveraging the original decoder to construct an efficient decoder for the rewritten diagram, including rewrites that keep the decoding problem intact up to isomorphisms (similar to \cite{fault-tolerant-transformations-of-spacetime-codes}).
This would constitute an alternative approach to preserving efficient decodability.

\section*{Acknowledgements}
We would like to thank Maximilian Rüsch, Alex Townsend-Teague and Peter-Jan Derks for useful discussions.
MS is supported by the German Academic Scholarship Foundation and employed at the Karlsruhe Institute of Technology.
LGS is supported through a Leverhulme-Peierls Fellowship at the University of Oxford, funded by grant no. LIP-2020-014. 
AK is supported by the Engineering and Physical Sciences Research Council grant number EP/Z002230/1, “(De)constructing quantum software (DeQS)”. 
BR thanks Simon Harrison for his generous support for the Wolfson Harrison UK Research Council Quantum Foundation Scholarship. 
BR is employed part-time by Quantinuum.
The wording in some sections of this paper has been refined using LLMs.

% \nocite{*}
\bibliographystyle{eptcs} % TODO: xor eptcsalpha, eptcsini, eptcsalphaini
\bibliography{references}
% TODO: follow referencing guidance from example.tex

\appendix
\crefalias{section}{appendix} % make cleverref label sections as appendices here

\section{Properties of Detector-Aware Rewrites}
\label{appendix:properties-of-da-rewrites}

\begin{proposition}[Detector-aware rewrites are compositional]
    \label{prop:da-compositional}
    Let $C(D)$ be a diagram with a Pauli web multiset $\mathcal{S}$. 
    Furthermore, let
    $(D, \mathcal{S}[D]) \underset{R}{\circeq} (E, \mathcal{T})$ be a
    detector-aware rewrite. 
    Then we can construct a detector-aware rewrite
    $(C(D), \mathcal{S}) \underset{R'}{\circeq} (C(E), \mathcal{S}')$ such that
    $\mathcal S'[C] = \mathcal S[C]$ and $\mathcal S'[E] = \mathcal T$.
\end{proposition}
\begin{proof}
    As $(D, \mathcal{S}[D]) \underset{R}{\circeq} (E, \mathcal{T})$ is a
    detector-aware rewrite, we know that $D \FaultEq E$.
    Therefore we have $C(D) \FaultEq C(E)$ by compositionality of fault
    equivalence.

    We can construct $\mathcal{S}'$ as given by \Cref{thm:applying-da-rewrites}
    and then $\mathcal S'[C] = \mathcal S[C]$ and $\mathcal S'[E] = \mathcal T$
    hold by construction.
    Furthermore, $\mathcal S'[C] = \mathcal S[C]$ implies that there exists a
    boundary-respecting coupling $R'$ between $\mathcal S$ and $\mathcal S'$.

    It remains to be shown that $\mathcal S'$ forms a local detector basis.
    To show that $\mathcal S'|_D$ contains a generating set, take any detecting
    Pauli web $P$ for $C(E)$ and transform it to a Pauli web $P'$ on $C(D)$,
    replacing the $E$ Pauli web coverage with $D$ Pauli web coverage such that
    $R(D) = E$.
    As $\mathcal S$ forms a local detector basis, $P'$ is generated by
    $P' = P_1' \cdots P_n'$ where $P_i' \in \mathcal S|_D$.
    Translating each $P_i'$ to $C(E)$ again by replacing $D$ coverage with $E$
    coverage following $R$, we obtain a Pauli web $P_1 \cdots P_n$ for $C(E)$
    that is equal to $P$ up to local detectors of $E$.
    This difference $P\cdot P_1 \cdots P_n$ can be generated by $\mathcal T|_D$, so $P$ is
    generated by $\mathcal S'|_D$.

    We can similarly show that $\mathcal S'|_D$ is independent by assuming a
    non-trivial set of detectors $P_1, \dots, P_n \in \mathcal S'|_D$ that
    compose to the trivial Pauli web (i.e. $P_1 \cdots P_n = Id$).
    Translating these from $C(E)$ to $C(D)$ following $R$ yields Pauli webs
    $P_1', \dots, P_n'$ that compose to a local detector of $C(D)$.
    We can reproduce this local detector from
    $P_{n+1} \cdots P_{m} \in \mathcal S[D]|_D$ so that
    $P_1' \cdots P_nP_{n+1} \cdots P_m = Id$.
    This contradicts the independence of $\mathcal S|_D$, so our assumption must
    be false and $\mathcal S'|_D$ is in fact independent. 
\end{proof}

\begin{proposition}[Detector-aware rewrites are transitive]
    Let $(D, \mathcal{S}) \underset{R}{\circeq} (E, \mathcal{T})$ and $(E, \mathcal{T}) \underset{R'}{\circeq} (F, \mathcal{V})$ be detector-aware rewrites. 
    Then $(D, \mathcal{S}) \underset{R' \circ R}{\circeq} (F, \mathcal{V})$ is a detector-aware rewrite.
\end{proposition}
\begin{proof}
    As $(D, \mathcal{S}) \underset{R}{\circeq} (E, \mathcal{T})$ and $(E, \mathcal{T}) \underset{R'}{\circeq} (F, \mathcal{V})$ are detector-aware rewrites, we have $D \FaultEq E$ and $E \FaultEq F$. 
    By transitivity of fault equivalence, this means that $D \FaultEq F$. 
    Also, as they are detector-aware rewrites, we know that $\mathcal{S}$ and $\mathcal{V}$ must form a basis of detecting regions in their respective diagrams. 
    Finally, the composition of two boundary-respecting couplings must also be a boundary-respecting coupling. 
    Therefore, $(D, \mathcal{S}) \underset{R' \circ R}{\circeq} (F, \mathcal{V})$ is a detector-aware rewrite. 
\end{proof}

\Cref{thm:applying-da-rewrites} tells us how to apply detector-aware rewrites to larger diagrams and how to construct the resulting Pauli web multiset from the previous one. 
However, our motivation for detector-aware rewrites was not to rewrite Pauli web \emph{multisets} but detector bases.
And indeed, detector bases are always mapped to detector bases by detector-aware rewrites.
\begin{corollary}[Detector-aware rewrites keep detector bases intact]
    \label{thm:da-keeps-det-basis-intact}
    When applying a detector-aware rewrite
    $(D, \mathcal{S}[D]) \underset{R}{\circeq} (E, \mathcal{T})$ to a diagram
    $C(D)$ with a detector basis $\mathcal S$, the resulting Pauli web multiset
    $\mathcal S'$ for $C(E)$ is a detector basis of $C(E)$.
\end{corollary}
\begin{proof}
    \Cref{prop:da-compositional} tells us that applying the detector-aware
    rewrite to $C(D)$ with $\mathcal S$ yields a Pauli web multiset
    $\mathcal S'$ for $C(E)$ such that
    $(C(D), \mathcal S) \circeq (C(E), \mathcal S')$ is a detector-aware rewrite
    with $\mathcal S[C] = \mathcal S'[C]$.
    Therefore $\mathcal S'$ must form a local detector basis.
    But as there exists a boundary-respecting coupling between $\mathcal S$ and
    $\mathcal S'$, all elements of $\mathcal S'$ must be local detectors, and
    $\mathcal S' = \mathcal S'|_D$ is a detector basis of $C(E)$.
\end{proof}

\section{Other matchability-preserving rewrites in the surface code}
\label{sec:appx:examples-continued}
In our derivation for a matchable, fault-tolerant implementation of the $Z^4$
plaquette measurement in the rotated surface code in \cref{sec:mp:example}, we
omitted all green Pauli webs for brevity.
Including all Pauli webs, the Pauli web coverage of the $Z^4$ measurement is as
follows:
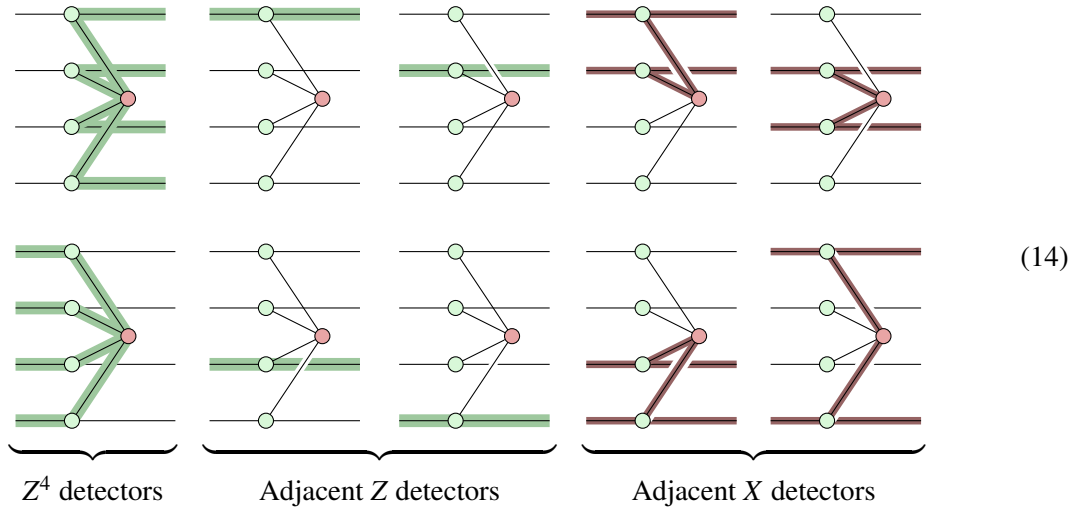
\begin{equation}
    \label{eq:plaquette-measurement-pauli-webs}
    \begin{array}{ccccc}
        \input{./figures/example/z-web-1.tikz}
            & \input{./figures/example/z-web-3.tikz}
            & \input{./figures/example/z-web-4.tikz}
            & \input{./figures/example/x-web-1.tikz}
            & \input{./figures/example/x-web-2.tikz}\\[48pt]
        \input{./figures/example/z-web-2.tikz}
            & \input{./figures/example/z-web-5.tikz}
            & \input{./figures/example/z-web-6.tikz}
            & \input{./figures/example/x-web-3.tikz}
            & \input{./figures/example/x-web-4.tikz}\\
        \multicolumn{1}{c}{$\upbracefill$}
            & \multicolumn{2}{c}{$\upbracefill$}
            & \multicolumn{2}{c}{$\upbracefill$}\\[4pt]
        \multicolumn{1}{c}{\text{$Z^4$ detectors}}
            & \multicolumn{2}{c}{\text{Adjacent $Z$ detectors}}
            & \multicolumn{2}{c}{\text{Adjacent $X$ detectors}}
    \end{array}
\end{equation}

We apply matchability-preserving rewrites to decompose the four-legged spider:
\begin{equation}
    \begin{aligned}
        2\times \input{./figures/example/zx-style-1-z-coloured.tikz} \quad&\cup\quad \input{./figures/example/zx-style-1-x-coloured.tikz}\\
        \stackrel{(r_4)}\circeq 2\times \input{./figures/example/zx-style-3-z-coloured.tikz} \quad&\cup\quad \input{./figures/example/zx-style-3-x-coloured.tikz}\\
        \stackrel{(r_{elim})}\circeq 2\times \input{./figures/example/zx-style-4-z-coloured.tikz} \quad&\cup\quad \input{./figures/example/zx-style-4-x-coloured.tikz}\\
        \stackrel{(r_{fuse})}\circeq 2\times \input{./figures/example/zx-style-5-z-coloured.tikz} \quad&\cup\quad \input{./figures/example/zx-style-5-x-coloured.tikz}\\
        \stackrel{(r_{elim})}\circeq 2\times \input{./figures/example/zx-style-6-z-coloured.tikz} \quad&\cup\quad \input{./figures/example/zx-style-6-x-coloured.tikz}\\
        \stackrel{(r_{elim})}\circeq 2\times \input{./figures/example/zx-style-7-z-coloured.tikz} \quad&\cup\quad \input{./figures/example/zx-style-7-x-coloured.tikz}\\
        \stackrel{(r_{fuse})}\circeq 2\times \input{./figures/example/zx-style-8-z-coloured.tikz} \quad&\cup\quad \input{./figures/example/zx-style-8-x-coloured.tikz}
    \end{aligned}
\end{equation}

Applying this rewrite to the full diagram leaves us with the following implementation:
\begin{equation}
    \label{eq:plaquette-measurement-zx-style-impl}
    \begin{array}{ccc}
        \scalebox{0.8}{\input{./figures/example/zx-style-9-z-web-3.tikz}}
            & \scalebox{0.8}{\input{./figures/example/zx-style-9-z-web-4.tikz}}
            & \scalebox{0.8}{\input{./figures/example/zx-style-9-z-web-1.tikz}}\\[50pt]
        \scalebox{0.8}{\input{./figures/example/zx-style-9-z-web-5.tikz}}
            & \scalebox{0.8}{\input{./figures/example/zx-style-9-z-web-6.tikz}}
            & \scalebox{0.8}{\input{./figures/example/zx-style-9-z-web-2.tikz}}\\[50pt]
        \scalebox{0.8}{\input{./figures/example/zx-style-9-x-web-1.tikz}}
            & \scalebox{0.8}{\input{./figures/example/zx-style-9-x-web-2.tikz}}
            & \multirow{2}{*}{\scalebox{0.8}{\input{./figures/example/zx-style-9-x-web-5.tikz}}}\\[50pt]
        \scalebox{0.8}{\input{./figures/example/zx-style-9-x-web-3.tikz}}
            & \scalebox{0.8}{\input{./figures/example/zx-style-9-x-web-4.tikz}}
    \end{array}
\end{equation}

In \cref{sec:mp:example}, we also pointed out that our framework allows us to
construct a matchable detector basis for Grans-Samuelsson et al.'s improved
plaquette measurement
implementation~\cite{improved-pairwise-measurement-based-surface-code}.
We derive this detector basis similar to the above construction by applying
matchability-preserving rewrites to the four-legged spider covered by Pauli webs
as visualised in~\cref{eq:plaquette-measurement-pauli-webs}.

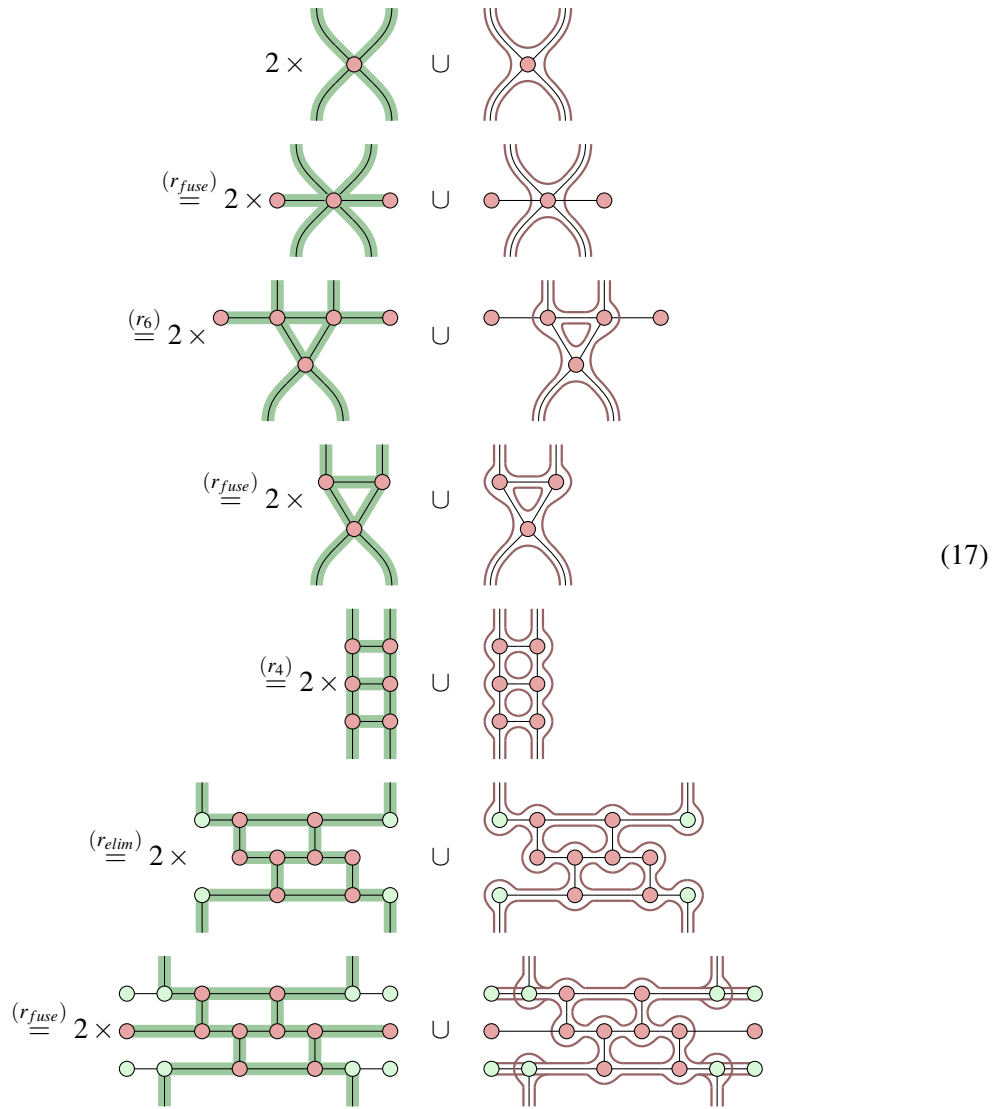
\begin{equation}
    \begin{aligned}
        2 \times \input{./figures/example/gs-style-01-z-coloured.tikz} &\quad\cup\quad \input{./figures/example/gs-style-01-x-coloured.tikz}\\
        \stackrel{(r_{fuse})}= 2 \times \input{./figures/example/gs-style-02-z-coloured.tikz} &\quad\cup\quad \input{./figures/example/gs-style-02-x-coloured.tikz}\\
        \stackrel{(r_6)}= 2 \times \input{./figures/example/gs-style-03-z-coloured.tikz} &\quad\cup\quad \input{./figures/example/gs-style-03-x-coloured.tikz}\\
        \stackrel{(r_{fuse})}= 2 \times \input{./figures/example/gs-style-04-z-coloured.tikz} &\quad\cup\quad \input{./figures/example/gs-style-04-x-coloured.tikz}\\
        \stackrel{(r_4)}= 2 \times \input{./figures/example/gs-style-05-z-coloured.tikz} &\quad\cup\quad \input{./figures/example/gs-style-05-x-coloured.tikz}\\
        \stackrel{(r_{elim})}= 2 \times \input{./figures/example/gs-style-06-z-coloured.tikz} &\quad\cup\quad \input{./figures/example/gs-style-06-x-coloured.tikz}\\
        \stackrel{(r_{fuse})}= 2 \times \input{./figures/example/gs-style-07-z-coloured.tikz} &\quad\cup\quad \input{./figures/example/gs-style-07-x-coloured.tikz}
    \end{aligned}
\end{equation}

Assuming that the measurement is repeated, we can ``roll over'' the three state
initialisations on the left of the diagram to the right of the diagram.
A formal treatment for such changes can be found in
\cite[Sec.~6]{floquetifying-stabiliser-codes}.
Note that in the diagrams below, we have drawn some of the red Pauli webs
separately since they cover parts which are only connected through the
surrounding diagram that is omitted here.

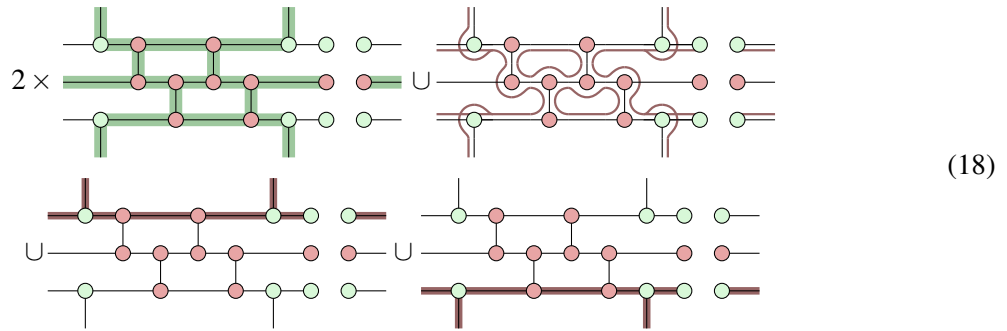
\begin{equation}
    \label{eq:gs-plaq-meas-first-rollover}
    \begin{gathered}
        2 \times \input{./figures/example/gs-style-08-z-coloured.tikz} \cup \input{./figures/example/gs-style-08-x-coloured.tikz}\\
        \cup \input{./figures/example/gs-style-08-x-coloured-alt1.tikz} \cup \input{./figures/example/gs-style-08-x-coloured-alt2.tikz}
    \end{gathered}
\end{equation}

The rollover allows us to merge the pairs of destructive measurement and
initialisation into repeated non-destructive measurements.
We state the corresponding matchability-preserving rewrite in the following
proposition.

\begin{proposition}[$r_{meas}$ preserves matchability]
    The following rewrite preserves CSS matchability.
    \begin{equation}
        \begin{tikzpicture}
	\begin{pgfonlayer}{nodelayer}
		\node [style=Z dot] (0) at (-0.5, 0) {};
		\node [style=Z dot] (1) at (0.5, 0) {};
		\node [style=none] (2) at (-1.5, 0) {};
		\node [style=none] (3) at (1.5, 0) {};
	\end{pgfonlayer}
	\begin{pgfonlayer}{edgelayer}
		\draw [style=IX W II] (2.center) to (0);
		\draw [style=IX W II] (1) to (3.center);
	\end{pgfonlayer}
\end{tikzpicture}
 \cup \begin{tikzpicture}
	\begin{pgfonlayer}{nodelayer}
		\node [style=Z dot] (0) at (-0.5, 0) {};
		\node [style=Z dot] (1) at (0.5, 0) {};
		\node [style=none] (2) at (-1.5, 0) {};
		\node [style=none] (3) at (1.5, 0) {};
	\end{pgfonlayer}
	\begin{pgfonlayer}{edgelayer}
		\draw [style=X Web] (2.center) to (0);
		\draw [style=X Web] (1) to (3.center);
	\end{pgfonlayer}
\end{tikzpicture}
 \circeq \input{./figures/example/meas-RHS-x-coloured.tikz}
    \end{equation}
\end{proposition}
\begin{proof}
    The given ZX rewrite was proven to be fault-equivalent
    in~\cite[Thm.~6.2]{floquetifying-stabiliser-codes}.
    It is detector-aware since both the LHS and RHS Pauli webs multisets form
    local detector bases, and there exists and obvious boundary-respecting
    coupling.
    We observe that both Pauli web multisets are CSS matchable and conclude that
    the rewrite preserves CSS matchability.
\end{proof}
We can generalise this to a triple measurement through repeated application:
\begin{equation}
    \begin{gathered}
        \input{./figures/example/triple-meas-1.tikz}\\
        \stackrel{(r_{meas})}\circeq \input{./figures/example/triple-meas-2-multi.tikz} \cup \input{./figures/example/triple-meas-2-single.tikz}\\
        \stackrel{(r_{elim})}\circeq \input{./figures/example/triple-meas-3-multi.tikz} \cup \input{./figures/example/triple-meas-3-single.tikz}\\
        \stackrel{(r_{fuse})}\circeq \input{./figures/example/triple-meas-4-multi.tikz} \cup \input{./figures/example/triple-meas-4-single.tikz}\\
        \stackrel{(r_{meas})}\circeq \input{./figures/example/triple-meas-5-multi.tikz} \cup \input{./figures/example/triple-meas-5-single.tikz}\\
        \circeq \begin{tikzpicture}
	\begin{pgfonlayer}{nodelayer}
		\node [style=Z dot] (0) at (-0.5, 0) {};
		\node [style=none] (2) at (-1.5, 0) {};
		\node [style=Z dot] (11) at (0.5, 0) {};
		\node [style=none] (12) at (1.5, 0) {};
	\end{pgfonlayer}
	\begin{pgfonlayer}{edgelayer}
		\draw [style=IX W II] (2.center) to (0);
		\draw [style=IX W II] (11) to (12.center);
	\end{pgfonlayer}
\end{tikzpicture}
 \cup \begin{tikzpicture}
	\begin{pgfonlayer}{nodelayer}
		\node [style=Z dot] (0) at (-0.5, 0) {};
		\node [style=none] (2) at (-1.5, 0) {};
		\node [style=Z dot] (11) at (0.5, 0) {};
		\node [style=none] (12) at (1.5, 0) {};
	\end{pgfonlayer}
	\begin{pgfonlayer}{edgelayer}
		\draw [style=X Web] (12.center) to (11);
		\draw [style=X Web] (0) to (2.center);
	\end{pgfonlayer}
\end{tikzpicture}

    \end{gathered}
\end{equation}

Applying these rewrites to the diagram from
\cref{eq:gs-plaq-meas-first-rollover} produces the following diagram.
\begin{equation}
    \input{./figures/example/gs-style-09-z-coloured.tikz} \cup \input{./figures/example/gs-style-09-z-coloured-alt.tikz} \cup \input{./figures/example/gs-style-09-x-coloured.tikz}
\end{equation}

Rolling over the diagram in the reverse direction gives us the implementation
from \cite{improved-pairwise-measurement-based-surface-code}.
\begin{equation}
    \input{./figures/example/gs-style-10-z-coloured.tikz} \cup \input{./figures/example/gs-style-10-z-coloured-alt.tikz} \cup \input{./figures/example/gs-style-10-x-coloured.tikz}
\end{equation}
We can obtain a full matchable basis by embedding the above diagrams and Pauli
webs into their context, similar to
\cref{eq:plaquette-measurement-zx-style-impl}.

\section{Rewrite variants}
\label{appendix:rewrite-variants}
We showed in \cref{sec:mp,sec:mp:generality} that the $4n$-legged spiders with a
flower-shaped Pauli web coverage can be decomposed with the
matchability-preserving $r_{4n}$ rewrite.
In this section, we present two similar rewrites, $r_{5n}$ and $r_{6n}$
which can be used to decompose $5n$ and $6n$-legged spiders under flower-shaped
Pauli web coverage, respectively.
While these rewrites are not strictly required for our compilation procedure to
be complete, picking these alternative rewrites can sometimes reduce the
resulting diagram size.
We leave a detailed comparison between these rewrites for future work. 

\begin{proposition}[$r_{5n}$ preserves matchability]
    The following family of rewrites preserves CSS matchability:
    \begin{equation}
        \begin{aligned}
            \input{./figures/circuit-extraction/r5n-blossom-LHS.tikz}
            \hspace{4pt}&\cup\hspace{4pt}
            2 \times \input{./figures/circuit-extraction/r5n-blossom-LHS-red.tikz}
            \\
            \stackrel{(r_{5n})}\circeq \quad
            \input{./figures/circuit-extraction/r5n-blossom-RHS.tikz}
            \hspace{4pt}&\cup\hspace{4pt}
            2 \times \input{./figures/circuit-extraction/r5n-blossom-RHS-red.tikz}
        \end{aligned}
    \end{equation}
    Here, the dots on the internal wires refer to bundles of $n$ wires with
    local detectors in between them.
\end{proposition}
\begin{proof}
    We prove this result in analogy to our proof for $r_{4n}$.
    For $n = 1$, the underlying ZX rewrite is the rewrite $r_5$ which was proven
    to be fault-equivalent in \cite[Thm~3.8]{floquetifying-stabiliser-codes}.
    Then $r_{5n}$ is the generalisation of $r_5$ where each of the wires is
    replaced by a bundle of $n$ wires.
    In analogy to our proof for $r_{4n}$, we argue that any undetectable error
    must cover either all or none of the wires in the wire bundles, and these
    bundled errors behave the same way as their simple counterparts in $r_5$.
    $r_{5n}$ is therefore a fault-equivalent rewrite.

    Furthermore, the rewrite along with the provided Pauli web multisets is
    clearly a detector-aware rewrite as each side of the equation forms a local
    detector basis, and there exists an obvious boundary-respecting coupling
    between them.
    Observing that both sides' Pauli web multisets are CSS matchable, we
    conclude that $r_{5n}$ is a matchability-preserving rewrite. 
\end{proof}

\begin{proposition}[$r_{6n}$ preserves matchability]
    The following family of rewrites preserves CSS matchability:
    \begin{equation}
        \begin{aligned}
            \input{./figures/circuit-extraction/r6n-blossom-LHS.tikz}
            \hspace{4pt}&\cup\hspace{4pt}
            2 \times \input{./figures/circuit-extraction/r6n-blossom-LHS-red.tikz}
            \\
            \stackrel{(r_{6n})}\circeq \quad
            \input{./figures/circuit-extraction/r6n-blossom-RHS.tikz}
            \hspace{4pt}&\cup\hspace{4pt}
            2 \times \input{./figures/circuit-extraction/r6n-blossom-RHS-red.tikz}
        \end{aligned}
    \end{equation}
    Here, the dots on the internal wires refer to bundles of $n$ wires with
    local detectors in between them.
\end{proposition}
\begin{proof}
    We prove the claim in analogy to $r_{4n}$ and $r_{5n}$.
    The underlying ZX rewrites for $n=1$ is as follows.
    \begin{equation}
        \input{./figures/circuit-extraction/r6-LHS.tikz}
        \quad = \quad
        \input{./figures/circuit-extraction/r6-RHS.tikz}
    \end{equation}
    We note that any $Z$ edge flip can be pushed out to either of the boundary
    wires.
    An undetectable fault must $X$-flip (or $Y$-flip) either none or exactly two
    of the internal wires because otherwise, the triangular local detector would
    be violated. 
    Any two such $X$ flips must be adjacent to a common $Z$ spider, so we can
    push them through that spider to obtain two $X$ flips on the boundary wires.
    \begin{equation}
        \input{./figures/circuit-extraction/r6-pushout-LHS.tikz}
        \quad = \quad
        \input{./figures/circuit-extraction/r6-pushout-RHS.tikz}
    \end{equation}
    This way, we can push all undetectable faults in the LHS and RHS of $r_6$
    to the boundary of the diagrams without increasing their weight, so by
    \cite[Prop.~A.4]{fault-tolerance-by-construction}, $r_6$ must be
    fault-equivalent.
    
    Similar to $r_{4n}$ and $r_{5n}$, we can generalise this fault equivalence
    to $r_{6n}$.
    We conclude the proof for $r_{6n}$ preserving matchability by observing that
    both the LHS and RHS Pauli web multisets in the claim each form local
    detector bases, they are in correspondence through the obvious
    boundary-respecting coupling and they are CSS matchable.
\end{proof}

\begin{proposition}
    \label{prop:appx:merging-stabilisers}
    Let $(D_1, \mathcal S_1) \circeq (D_2, \mathcal S_2)$ be a
    matchability-preserving rewrite, let $B_i$ be the boundary edges of $D_i$
    and let $P_i,P_i' \in \mathcal S_i$ with $P_i \neq P_i'$ (for $i = 1, 2$)
    be Pauli webs with non-trivial action on the boundary such that
    $P_1[B_1] = P_2[B_2] \neq P_1'[B_1] = P_2'[B_2]$.

    Replacing $P_i, P_i'$ with $P_iP_i'$ in $\mathcal S_i$ (for $i = 1, 2$)
    yields another matchability-preserving rewrite
    $(D_1, \mathcal S_1') \circeq (D_2, \mathcal S_2')$.
\end{proposition}
\begin{proof}
    Following \cref{prop:submultisets}, removing $P_1, P_1'$ from $\mathcal S_1$
    yields a matchability-preserving submultiset-rewrite
    $(D_1, \tilde{\mathcal S_1}) \circeq (D_2, \tilde{\mathcal S_2})$ where
    $P_2, P_2'$ are removed from $\mathcal S_2$ too.
    Then for $\mathcal S_i' := \tilde{\mathcal S_i} \cup P_iP_i'$ (for
    $i = 1, 2$), $(D_1, \mathcal S_1') \circeq (D_2, \mathcal S_2)$ is a
    detector-aware rewrite too since $P_1P_1'[B_1] = P_2P_2'[B_2]$.
    We observe that $P_iP_i'$ can have a nontrivial action on at most as many
    edges as $P_i$ and $P_i'$ combined (for $i = 1, 2$), so we conclude that
    $(D_1, \mathcal S_1') \circeq (D_2, \mathcal S_2)$ is a
    matchability-preserving rewrite.
\end{proof}
As an important special case, \cref{prop:appx:merging-stabilisers} allows us to
combine two adjacent Pauli webs in a flower-shaped Pauli web coverage to ensure
a single wire is covered trivially by the Pauli web multiset.
This allows us to transform $r_{4n}$, $r_{5n}$ and $r_{6n}$ to decompose spiders
resulting from \cref{rewrite:unfuse-floral}.

\end{document}